\documentclass[12pt]{iopart}
\usepackage{latexsym}
\usepackage[T1]{fontenc}
\usepackage[utf8x]{inputenc}
\usepackage{times}
\usepackage{float}
\usepackage{amsthm}
\usepackage{amssymb}
\usepackage{amsfonts}
\usepackage{amstext}
\usepackage{graphicx}
\usepackage{array}
\usepackage{multirow}
\usepackage{subfig}
\usepackage{float}
\usepackage{indentfirst}
\usepackage{enumitem}
\usepackage{bbm}
\usepackage{multirow}

\newcommand{\bra}[1]{\mbox{$\langle #1 |$}}
\newcommand{\ket}[1]{\mbox{$| #1 \rangle$}}
\newcommand{\bk}[2]{\ensuremath{\langle #1 | #2 \rangle}}
\newcommand{\kb}[2]{\ensuremath{| #1 \rangle\!\langle #2 |}}

\newcommand{\sr}{\rule[-0.1cm]{0pt}{0.55cm}}

\newtheorem{theorem}{Theorem}
\newtheorem{lemma}[theorem]{Lemma}

\newtheorem{definition}[theorem]{Definition}
\newtheorem{remark}{Remark}

\begin{document}

\title{Critical points of the linear entropy for pure $L$-qubit states}

\author{Tomasz Maci\c{a}\.{z}ek}
\address{Center for Theoretical Physics, Polish Academy of
Sciences, Al. Lotnik\'ow 32/46, 02-668 Warszawa, Poland}
\address{Faculty of Physics, University of Warsaw, ul. Pasteura 5, 02-093 Warszawa, Poland}
\author{Adam Sawicki}
\address{Center for Theoretical Physics, Polish Academy of
Sciences, Al. Lotnik\'ow 32/46, 02-668 Warszawa, Poland}
\address{School of Mathematics, University of Bristol, University Walk, Bristol BS8 1TW, UK}
\address{Center for Theoretical Physics, Massachusetts Institute of Technology, 77 Massachusetts Avenue, Cambridge, MA 02139, USA}
\date{\today}

\begin{abstract}
We present a substantially improved version of the method proposed
in \cite{SOK12, SOK14} for finding critical points of the linear entropy
for $L$-qubit system. The new approach is based on the correspondence
between momentum maps for abelian and non-abelian groups, as described
in \cite{K84-1}. The proposed method can be implemented numerically much easier than the previous one.
\end{abstract}

\maketitle

\section{Introduction }

In this paper we consider pure states of the quantum system consisting of $L$ qubits. Throughout, we assume that considered qubits are isolated. Such setting can be realised as, for example, a set of dual-rail photons or two-level atoms. In the first case, a qubit is encoded by a single photon occupying two orthogonal modes, while in the second case the role of a qubit is played by an electron in an atom with two non-degenerate energy levels. Considering a number of isolated qubits means taking into account only the systems of noninteracting, distinguishable particles. Continuing the examples above, a possible realization would be a set of  two-mode noninteracting waveguides or a set of noninteracting atoms respectively. The hamiltonian evolution of such system is described by a group of local unitaries, $K=SU(2)^{\times L}$, acting independently on each qubit, i.e. for $k=(k_1,\dots,k_L)\in K$
\begin{equation*}
\fl k.\ket{\phi_1}\otimes\ket{\phi_2}\otimes\dots\otimes\ket{\phi_L}=\ket{k_1.\phi_1}\otimes\ket{k_2.\phi_2}\otimes\dots\otimes\ket{k_L.\phi_L},\ k_i\in SU(2).
\end{equation*}  
In the above expression $\ket{\phi_i}\in\mathbb{C}^2$, as the Hilbert space, $\mathcal{H}$, of the considered $L$-qubit system is equal to $\left(\mathbb{C}^2\right)^{\otimes L}$.
Moreover, the only measurements which are available are restricted to the local one-qubit measurements and the allowed operations are, in general, the stochastic local operations assisted by classical communication (SLOCC).  In other words, the SLOCC operations are represented by the local linear operators. However, in our work we restrict ourselves to considering only invertible SLOCC operations, namely to the action of group $G=SL(2,\mathbb{C})^{\times L}$. 

Another consequence of the partition into $L$ distinguishable noninteracting parties is that one does not need to know exactly the state $\phi\in\mathcal{H}$ in order to calculate expectation values of the available local measurements. In fact, from a single qubit perspective, the necessary and sufficient information about $\phi\in\mathcal{H}$ is encoded in one-qubit reduced density matrices $\{\rho_{1}(\phi),\ldots,\rho_{L}(\phi)\}$. The knowledge of the reduced one-qubit density matrices allows one to define the total variance of state $\phi$ \cite{SOK12, SOK14, klyachko}, or equivalently the linear entropy at the state $\phi$, $\mathrm{E}(\phi)$, which we describe in more detail in the next section. The main importance of the linear entropy of the reduced matrices is that it can be used for quantifying entanglement in the system. The arguments for such an interpretation are the following. The linear entropy of a density matrix, $\rho$, is also called the impurity of a state and is given by the formula $E(\rho)=1-\Tr(\rho^2)$. For a pure state $\rho^2=\rho$ and therefore the linear entropy is equal to zero. However, performing the reduction of a density matrix of a pure state by tracing over a subsystem usually results with a mixed state. The mixedness of the reduced density matrices of a pure state can arise only from the entanglement between the subsystems. Therefore, one can conjecture that the linear entropy of the reduced density matrices can be in this case a good measure of the entanglement between the subsystems. Such observations have been already made and used in various contexts, for example in investigating entanglement in multi-qubit-cavity systems \cite{example_cavity}, in two-fermion systems \cite{example_fermions}  or in ions \cite{example_ions}. In this paper, for the system of $L$ qubits, we aim to give a systematic and tractable method for finding critical points of the linear entropy. A critical point is a state at which the derivative of the linear entropy is equal to zero. Moreover, if $\phi$ is such a critical point then the linear entropy restricted to the set $G.\phi:=\{g.\phi,\ g\in SL(2,\mathbb{C})^{\times L}\}$ attains its maximum value \cite{Ness} at $\phi$. Hence, by our interpretation of the linear entropy, such critical points are the locally maximally entangled states. As was shown in \cite{SOK12, SOK14} critical points of $\mathrm{E}(\phi)$ also play an essential role in the classification of states with respect to SLOCC operations as they parametrize quantum states that are asymptotically equivalent, i.e. can be transformed into one another by applying an infinite sequence of SLOCC operations. In \cite{SOK12, SOK14} one of us \footnote{with M. Ku\'s and M. Oszmaniec} outlined the method which in principle allows finding critical points of $E(\phi)$ for arbitrary system of both distinguishable and indistinguishable particles. To illustrate the method some low dimensional examples were calculated. Here we substantially improve the method proposed in \cite{SOK12, SOK14}. Our improvement is based on the general theory of the critical points of the momentum map as presented in \cite{K84-1}. In particular it is based on the interplay between momentum maps for abelian and non-abelian groups. The general construction of the momentum map can be found, for example in \cite{GS}. The motivation for considering the momentum map is the description of symmetries in a given system. For example, in a classic mechanical system whose Hamiltonian is invariant under the group of translations, the momentum map is a map which assigns to a particle's trajectory its classical momentum. Because the space of pure quantum states has the same properties as the classical phase space, namely both are symplectic manifolds, one can use the symmetry of preserving the norm of $\phi\in\mathcal{H}$ by the group of local unitaries in order to construct the momentum map.

The paper is organized as follows. First, in section \ref{sec:The-linear-entropy}
we review some useful facts concerning reduced one-qubit density matrices,
linear entropy and its critical points. In particular we briefly describe
the method proposed in \cite{SOK12, SOK14}. In section \ref{crit-mu} we discuss the connection between the linear entropy and the momentum map. In particular we consider two momentum maps, for the action of local unitary group and its abelianization. Then, in section 4 we describe the implementation for $L$ qubits.

\section{The linear entropy and its critical points\label{sec:The-linear-entropy}}

In this section we shortly review the relevant background informations
concerning reduced one-qubit density matrices, the linear entropy
and its critical points.

Let $\mathcal{H}=\left(\mathbb{C}^{2}\right)^{\otimes L}$ be the
Hilbert space of $L$-qubit system. Recall that reduced one-qubit
density matrices, $\{\rho_{i}(\phi)\}_{i=1}^{L}$, are the $2\times2$
density matrices that satisfy

\begin{eqnarray}
\label{eq:reduced_def}
\frac{\bk{\phi}{A_{1}\otimes\mathbbm{1}\otimes\ldots\otimes\mathbbm{1}+\ldots+\mathbbm{1}\otimes\ldots\otimes\mathbbm{1}\otimes A_{L}|\phi}}{\bk{\phi}{\phi}}=\sum_{i=1}^{L}\Tr(A_{i}\rho_{i}),
\end{eqnarray}
for an arbitrary set of $2\times2$ hermitian matrices $\{A_{i}\}_{i=1}^L$ which represents
an observable measured on the $i$-th qubit. In order to clarify the notion of a reduced matrix, let us also give an alternative, basis - dependent definition, incorporating the operation of the partial trace of a matrix. We will focus on a bipartite system, as the definition for arbitrary multipartite systems has a straightforward extension. If $\rho_{AB}$ is a density matrix of a two partite system and $\{\ket{i}_B\}$ is an orthonormal basis of the Hilbert space of subsystem $B$, then the reduced density matrix of subsystem $A$ is defined by
\begin{equation*}
\rho_A=\Tr_B(\rho_{AB})=\sum_{i} {}_B\bra{i}\rho_{AB}\ket{i}{}_B.
\end{equation*}
In the case of more particles, in order to find one-particle reduced density matrices, one has to perform the above operation repeatedly with respect to each subsystem. Such a reduced matrix can be interpreted as an averaged state of the chosen particle, where the average is taken over all remaining particles.

\noindent Let us next define two functions, which we will denote by $\mu$ and $\Psi$ that act on a space of quantum states and return a set of $2\times2$ traceless hermitian matrices. In a group theoretical language, maps $\mu$ and $\Psi$ assign to a quantum state a set of elements of the Lie algebra $\mathfrak{su}(2)$ multiplied by the imaginary unit. The above mentioned hermitian matrices are the shifted one-qubit reduced density matrices, namely
\begin{eqnarray}
\mu(\phi)=\{\rho_{1}(\phi)-\frac{1}{2}\mathbbm{1},\rho_{2}(\phi)-\frac{1}{2}\mathbbm{1},\ldots,\rho_{L}(\phi)-\frac{1}{2}\mathbbm{1}\},\nonumber \\
\Psi(\phi)=\{\tilde{\rho}_{1}(\phi)-\frac{1}{2}\mathbbm{1},\tilde{\rho}_{2}(\phi)-\frac{1}{2}\mathbbm{1},\ldots,\tilde{\rho}_{L}(\phi)-\frac{1}{2}\mathbbm{1}\},\label{eq:set-of-matrices}
\end{eqnarray}
where each $\tilde{\rho}_{i}$ is a diagonal $2\times2$ matrix whose
diagonal elements are given by the increasingly ordered spectrum of
$\rho_{i}$. The map $\Psi$ assigns to a state $\phi$
the collection of the shifted spectra of its (shifted) one-qubit reduced density matrices. Let $\{p_{i},\,1-p_{i}\}$
be the ordered spectrum of $\rho_{i}$, that is, $\frac{1}{2}\geq p_{i}\geq0$
and $p_{i}\leq1-p_{i}$. The shifted spectrum, i.e. the spectrum of
$\rho_{i}(\phi)-\frac{1}{2}I$ is given by $\{-\lambda_{i},\lambda_{i}\}$
where $0\leq\lambda_{i}=\frac{1}{2}-p_{i}\leq\frac{1}{2}$ . Under
these assumptions the image $\Psi(\mathcal{H})$ is known to be a
convex polytope, defined by the set of polygonal inequalities \cite{HSS03}

\begin{eqnarray}
\forall_{i}\ \frac{1}{2}-\lambda_{i}\leq\sum_{i\neq j}\left(\frac{1}{2}-\lambda_{i}\right)\,.\label{nier1}
\end{eqnarray}
The images of $\Psi$ in the case of two and three qubits are shown in Figure \ref{fig:polytopes}. Note that for $L=2$ inequalities (\ref{nier1}) imply that $\lambda_1=\lambda_2$, which is a manifestation of a well known fact that for a two-partite system, the spectra of the one-particle reduced density matrices are identical. The line from $\mathrm{v}_{GHZ}$ to $\mathrm{v}_{SEP}$ on Figure \ref{fig:polytopes}(a) can be parametrized using the Schmidt decomposition. Namely, any two-qubit state can be can be written in a proper basis as
\begin{equation*}
\ket{\phi(\theta)}=\sin\frac{\theta}{2}\ket{00}+\cos\frac{\theta}{2}\ket{11},\ \theta\in\left[0,\frac{\pi}{2}\right].
\end{equation*}
For such a state we have $\lambda_1=\lambda_2=\cos^2\frac{\theta}{2}-\frac{1}{2}=\frac{1}{2}\cos\theta$ and $\ket{\phi(0)}=\ket{11}$, while $\ket{\phi(\frac{\pi}{2})}=\frac{1}{\sqrt{2}}\left(\ket{00}+\ket{11}\right)$. In the case of $L=3$, the polytope is three-dimensional and is spanned on $5$ vertices.
\begin{figure}[h]
\centering
\includegraphics[width=.9\textwidth]{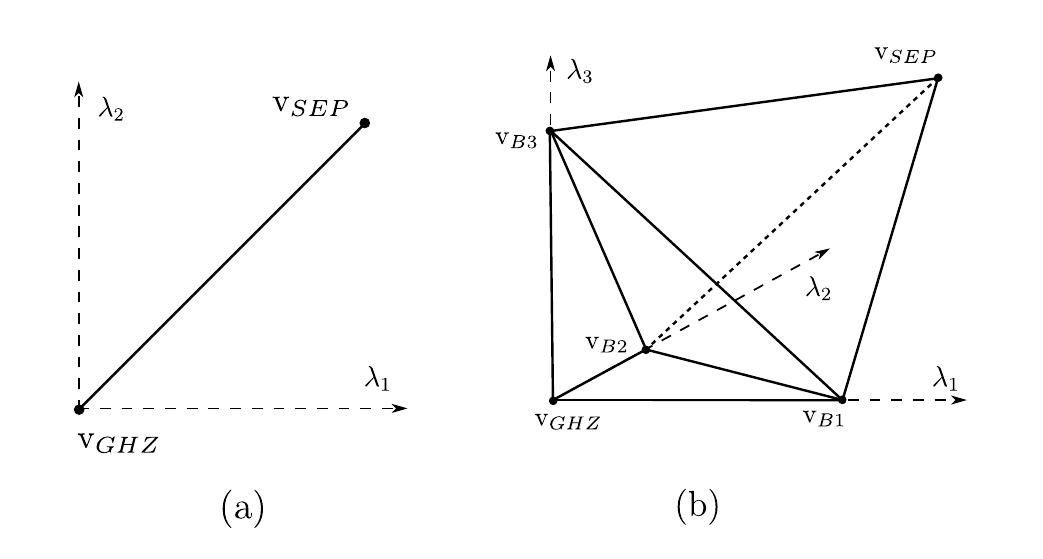}
\caption{The images of map $\Psi$ for (a) two and (b) three qubits. Point $\mathrm{v}_{GHZ}$ is the image of $\ket{GHZ}=\frac{1}{\sqrt{2}}\left(\ket{00}+\ket{11}\right)$ or $\ket{GHZ}=\frac{1}{\sqrt{2}}\left(\ket{000}+\ket{111}\right)$ respectively, while points $\mathrm{v}_{Bi}$ correspond to biseparable states.}
\label{fig:polytopes}
\end{figure}

\noindent Note also that both $\mu$ and $\Psi$ are not sensitive to changes of the global phase and norm of state $\phi\in\mathcal{H}$ and as such are defined on the complex projective space $\mathbb{P}(\mathcal{H})$ rather than $\mathcal{H}$. The complex projective space is a set of equivalence classes of vectors $v\in\mathcal{H}-\{0\}$, where all vectors belonging to the same complex line are identified. In other words, $v_1\sim v_2$ if and only if $v_2=cv_1$, $c\in\mathbb{C}$. The notion of a complex projective space emerges in a natural way while considering quantum states, because a pure quantum state can be viewed as a representative of some equivalence class of relation $\sim$. Space $\mathbb{P}(\mathcal{H})$ has some important geometric properties, which we will use in the following sections.

\paragraph*{The linear entropy} In the following paragraph, we review some physical aspects of the linear entropy. One way to quantify entanglement of a pure multipartite state is to calculate the von Neumann entropy of the reduced density matrices. This is because the von Neumann entropy of a quantum state is nonzero only if the considered state is mixed. On the other hand, the only reason for a reduced density matrix of a pure state to be mixed is the presence of entanglement between the particles. Therefore, one may expect that the sum of von Neumann entropies of all reduced density matrices quantifies the amount of entanglement in the system. Let us now return to the case of $L$ qubits. The spectrum of the $i$th one-qubit reduced density matrix is given by $\{p_i^{(0)},p_i^{(1)}\}$, so the von Neumann entropy is of the form
\begin{equation*}
E_N(\rho_i)=-p_i^{(0)}\log p_i^{(0)}-p_i^{(1)}\log p_i^{(1)}.
\end{equation*}   
Approximating each of the logarithms by $p_i^{(j)}-1,\ j=0,1$ (see Figure \ref{fig:approximation} for the plot of the approximation versus the original function) and using the fact that $p_i^{(0)}+p_i^{(1)}=1$, one obtains
\begin{equation*}
\fl E_N(\rho_i)\approx\sum_{j=0}^1 p_i^{(j)}\left(1- p_i^{(j)}\right)=1-\sum_{j=0}^1 \left(p_i^{(j)}\right)^2=1-\Tr\left(\rho_i^2\right).
\end{equation*}

\begin{figure}[h]
\centering
\includegraphics[width=.7\textwidth]{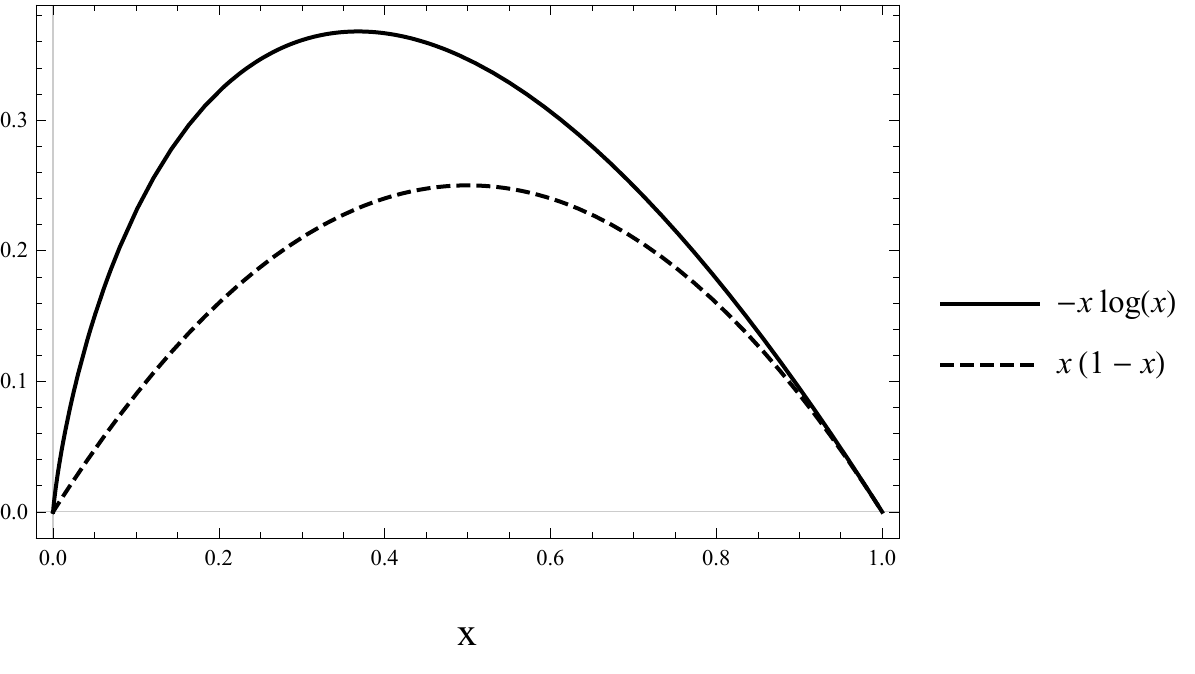}
\caption{Plot of the approximation used for linearizing the von Neumann entropy. The approximated function was $-x\log(x)$.}
\label{fig:approximation}
\end{figure}
Summing the approximations of the von Neumann entropy for all one-qubit reduced density matrices and multiplying the result by factor $\frac{1}{L}$, one gets the formula for the linear entropy of a state $\phi$ \cite{SOK12, SOK14}
\begin{eqnarray}
E(\phi)=1-\frac{1}{L}\sum_{i=1}^{L}\Tr(\rho_{i}^{2}(\phi)).\label{eq:linear-entropy}
\end{eqnarray}
Another way to interpret the linear entropy as a quantification of entanglement in a quantum system of many particles, is to consider quantum fluctuations as a manifestation of entanglement. Indeed, the more entangled a quantum state is, the more inaccurate the results of measurements performed on this state are. This approach allows one to define the linear entropy as an accumulated variance of some set of observables, i.e. a quantity that can be \textit{a priori} measured. To prove this statement, let us define after Klyachko \cite{Klyachko2007} the total variance of a pure state
\begin{equation}
\mathrm{Var}(\phi):=\frac{1}{\bk{\phi}{\phi}}\left(\sum_{i=1}^{\dim K}\bra{\phi}\xi_i^2\ket{\phi}-\frac{1}{\bk{\phi}{\phi}}\sum_{i=1}^{\dim K}\left[\bra{\phi}\xi_i\ket{\phi}\right]^2\right),
\end{equation}
where $\{\xi_i\}_{i=1}^{\dim K}$ is an orthonormal basis of the traceless local hermitian operators acting on $\mathcal{H}$. Note that the space of (traceless) local hermitian operators is in our setting equal to the Lie algebra of group K, modulo the multiplication by the imaginary unit, therefore their dimensions are the same. An exemplary basis of such a hermitian operators can be constructed from the Pauli matrices, $\sigma_{x,y,z}$, namely
\[\mathcal{X}_{k}=\mathbbm{1}\otimes\mathbbm{1}\otimes...\otimes\sigma_{x}\otimes\mathbbm{1}\otimes...\otimes\mathbbm{1},\]
\[\mathcal{Y}_{k}=\mathbbm{1}\otimes\mathbbm{1}\otimes...\otimes\sigma_{y}\otimes\mathbbm{1}\otimes...\otimes\mathbbm{1},\]
\[\mathcal{Z}_{k}=\mathbbm{1}\otimes\mathbbm{1}\otimes...\otimes\sigma_{z}\otimes\mathbbm{1}\otimes...\otimes\mathbbm{1}.\]
Using formula (\ref{eq:reduced_def}), one has that
\begin{equation}
\label{eq:sigma_x}
\frac{1}{\bk{\phi}{\phi}^2}\left[\bra{\phi}\mathcal{X}_{k}\ket{\phi}\right]^2=\left(\Tr[\rho(\phi)\mathcal{X}_{k}]\right)^2=\left(\Tr[\rho_k(\phi)\sigma_x]\right)^2.
\end{equation}
On the other hand, every $\rho_i$ can be written as a linear combination of Pauli matrices, $\rho_i=\frac{1}{2}\left(\mathbbm{1}+n_x\sigma_x+n_y\sigma_y+n_z\sigma_z\right)$, where $n_l=\Tr(\rho_i\sigma_l)$. Therefore, equation (\ref{eq:sigma_x}) becomes
\begin{equation*}
\frac{1}{\bk{\phi}{\phi}^2}\left[\bra{\phi}\mathcal{X}_{k}\ket{\phi}\right]^2=n_x^2
\end{equation*}
and analogical result applies to $\mathcal{Y}_{k}$ and $\mathcal{Z}_{k}$. Moreover, one can check by straightforward calculation that $\Tr[\rho_k(\phi)^2]=\frac{1}{2}\left(1+n_x^2+n_y^2+n_z^2\right)$. Hence, 
\begin{equation*}
\fl\frac{1}{\bk{\phi}{\phi}^2}\left(\left[\bra{\phi}\mathcal{X}_{k}\ket{\phi}\right]^2+\left[\bra{\phi}\mathcal{Y}_{k}\ket{\phi}\right]^2+\left[\bra{\phi}\mathcal{Z}_{k}\ket{\phi}\right]^2\right)=n_x^2+n_y^2+n_z^2=2\Tr[\rho_k(\phi)^2]-1.
\end{equation*}
Finally, using the fact that $\sigma_l^2=\mathbbm{1}$, one obtains the formula for the total variance of a pure state of $L$ qubits
\begin{equation}
\mathrm{Var}(\phi)=4L-2\sum_{i=1}^L\Tr[\rho_i(\phi)^2]=2L[E(\phi)+1]
\end{equation}
For a more detailed discussion, see \cite{SOK14}.

\noindent Having the above justification for considering the linear entropy as a measure of entanglement, the search for critical points of function $E(\cdot)$ would be of a great importance. The state $\phi$ is a critical point of $E(\cdot)$ if and only if
the differential of $E$ vanishes at $\phi$, $dE(\phi)=0$. In fact, function $E(\cdot)$ restricted to the $G$-orbit through a critical point, attains its maximum value at this point \cite{SOK14,Ness}. In order to characterize the critical points of $E$ we define the local operators constructed from $\{\rho_{i}(\phi)\}_{i=1}^{L}$ 

\begin{eqnarray}
\fl \widehat{\mu}(\phi)=\left(\rho_{1}(\phi)-\frac{1}{2}\mathbbm{1}\right)\otimes\mathbbm{1}\otimes\ldots\otimes\mathbbm{1}+\ldots+\mathbbm{1}\otimes\ldots\otimes\mathbbm{1}\otimes\left(\rho_{L}(\phi)-\frac{1}{2}\mathbbm{1}\right),\label{eq:operators}\\
 \fl \widehat{\Psi}(\phi)=\left(\tilde{\rho}_{1}(\phi)-\frac{1}{2}\mathbbm{1}\right)\otimes\mathbbm{1}\otimes\ldots\otimes\mathbbm{1}+\ldots+\mathbbm{1}\otimes\ldots\otimes\mathbbm{1}\otimes\left(\tilde{\rho}_{L}(\phi)-\frac{1}{2}\mathbbm{1}\right)\,.\nonumber 
\end{eqnarray}
As was shown in \cite{SOK12, SOK14}, $\phi$ is critical if and only if 

\begin{eqnarray}
\widehat{\mu}(\phi)\phi=\lambda\phi.\label{eq:eigen-proble}
\end{eqnarray}
The eigenproblem (\ref{eq:eigen-proble}) is typically hard to solve
as the matrix $\widehat{\mu}(\phi)$ depends in the non-linear way
on $\phi$. Note, however, that the linear entropy is $K$-invariant
function, that is, for any local unitary $U\in K$ 

\begin{eqnarray*}
E(U\phi)=1-\frac{1}{L}\sum_{i=1}^{L}\mathrm{tr}(\rho_{i}^{2}(U\phi))=1-\frac{1}{L}\sum_{i=1}^{L}\Tr(U\rho_{i}^{2}(\phi)U^{\dagger})=E(\phi).
\end{eqnarray*}
Therefore, the critical points of $E(\cdot)$ form $K$-orbits. In other words, if $\phi$ is a critical point, then the whole set $\mathcal{O}_\phi=\{\psi:\ \psi=U\phi, U\in K\}$, called a $K$-orbit through point $\phi$, is critical. Moreover,
 when we act with local unitaries, $U=(U_1,\dots,U_L)\in K$, on a state $\phi$, each one-qubit reduced density matrix transforms as follows
\begin{equation*}
\rho_i(U\phi)=U_i\rho_i(\phi)U_i^\dagger=U_i\rho_i(\phi)U_i^{-1}.
\end{equation*}
Matrices from $U$ can be chosen to diagonalize each $\rho_i(\phi)$, hence each $K$-orbit in $\mathbb{P}(\mathcal{H})$ corresponds to one point in the polytope $\Psi(\mathbb{P}(\mathcal{H}))$.
Thus the critical reduced one-qubit density matrices correspond to some subset of points in $\Psi(\mathbb{P}(\mathcal{H}))$. In fact, this subset contains a finite number of points \cite{K84-1}. As was proposed in \cite{SOK12, SOK14} its determination can be reduced to the following procedure

\paragraph{The method for finding critical reduced one-qubit density matrices
as proposed in \cite{SOK12, SOK14}}
\begin{enumerate}
\item For $\alpha\in\Psi(\mathbb{P}(\mathcal{H}))$, where $\alpha=\{\alpha_{1},\ldots,\alpha_{L}\}$
construct the operator 
\begin{eqnarray*}
\alpha=\alpha_{1}\otimes\mathbbm{1}\otimes\ldots\otimes\mathbbm{1}+\ldots+I\otimes\ldots\otimes\mathbbm{1}\otimes\alpha_{L},
\end{eqnarray*}
 where $\alpha_{i}=\mathrm{diag}(-\lambda_{i},\lambda_{i})$, $\lambda_{i}\in[0,\frac{1}{2}]$
and inequalities (\ref{nier1}) are satisfied.
\item Find the eigenspaces of operator $\alpha$
\item For each eigenspace verify if it contains state $\phi$ such that
$\mu(\phi)=\alpha$. If this is the case then $\alpha\in\Psi(\mathbb{P}(\mathcal{H}))$
is critical.
\end{enumerate}
In \cite{SOK12, SOK14} the algorithm was applied to find the critical points
up to three qubits. For two qubits there are only two critical points,
i.e. the GHZ and the separable states. For three qubits we have the GHZ, separable,
three bi-separable and the W states (see table 1). In the case of four qubits there are already nine critical sets of one-qubit density matrices. They are listed in table \ref{tab:4qubit_old} with exemplary critical states. As pointed out in \cite{SOK12, SOK14} they are in 1-1 correspondence
with SLOCC classes. 

\begin{table}[h]
\centering
\begin{tabular}{|c|c|c|}
\hline 
Critical $\alpha\in\Psi(\mathbb{P}(\mathcal{H}))$ & State & E($\phi$)\\
\hline 
\hline 
$\left(\begin{array}{cc}
-\frac{1}{2} & 0\\
0 & \frac{1}{2}
\end{array}\right),\,\left(\begin{array}{cc}
-\frac{1}{2} & 0\\
0 & \frac{1}{2}
\end{array}\right),\,\left(\begin{array}{cc}
-\frac{1}{2} & 0\\
0 & \frac{1}{2}
\end{array}\right)$  & Sep & 0\\
\hline 
$\left(\begin{array}{cc}
-\frac{1}{2} & 0\\
0 & \frac{1}{2}
\end{array}\right),\,\left(\begin{array}{cc}
0 & 0\\
0 & 0
\end{array}\right),\,\left(\begin{array}{cc}
0 & 0\\
0 & 0
\end{array}\right)$ &  BiSep & $\frac{1}{3}$\\
\hline 
$\left(\begin{array}{cc}
-\frac{1}{6} & 0\\
0 & \frac{1}{6}
\end{array}\right),\,\left(\begin{array}{cc}
-\frac{1}{6} & 0\\
0 & \frac{1}{6}
\end{array}\right),\,\left(\begin{array}{cc}
-\frac{1}{6} & 0\\
0 & \frac{1}{6}
\end{array}\right)$ &  W & $\frac{4}{9}$\\
\hline 
$\left(\begin{array}{cc}
0 & 0\\
0 & 0
\end{array}\right),\,\left(\begin{array}{cc}
0 & 0\\
0 & 0
\end{array}\right),\,\left(\begin{array}{cc}
0 & 0\\
0 & 0
\end{array}\right)$ &  GHZ & $\frac{1}{2}$\\
\hline 
\end{tabular}
\caption{Three qubits critical one-qubit matrices. The listed states are $\ket{\textrm{BiSep}}=\frac{1}{\sqrt{2}}\ket{1}\otimes (\ket{00}+\ket{11})$, $\ket{\textrm{W}}=\frac{1}{\sqrt{3}}(\ket{110}+\ket{101}+\ket{011})$ and $\ket{\textrm{GHZ}}=\frac{1}{\sqrt{2}}(\ket{000}+\ket{111})$.}
\label{tab:3qubit_old}
\end{table}

\begin{table}[h]
\centering
\begin{tabular}{|c|c|c|}
\hline 
Critical $\alpha\in\Psi(\mathbb{P}(\mathcal{H}))$ & State & E($\phi$)\\
\hline 
\hline 
$\left(\begin{array}{cc}
-\frac{1}{2} & 0\\
0 & \frac{1}{2}
\end{array}\right),\,\left(\begin{array}{cc}
-\frac{1}{2} & 0\\
0 & \frac{1}{2}
\end{array}\right),\,\left(\begin{array}{cc}
-\frac{1}{2} & 0\\
0 & \frac{1}{2}
\end{array}\right)$  & Sep & 0\\
\hline 
$\left(\begin{array}{cc}
-\frac{1}{2} & 0\\
0 & \frac{1}{2}
\end{array}\right),\,\left(\begin{array}{cc}
-\frac{1}{2} & 0\\
0 & \frac{1}{2}
\end{array}\right),\,\left(\begin{array}{cc}
0 & 0\\
0 & 0
\end{array}\right),\,\left(\begin{array}{cc}
0 & 0\\
0 & 0
\end{array}\right)$ &  TriSep & $\frac{1}{4}$\\
\hline 
$\left(\begin{array}{cc}
-\frac{1}{6} & 0\\
0 & \frac{1}{6}
\end{array}\right),\,\left(\begin{array}{cc}
-\frac{1}{6} & 0\\
0 & \frac{1}{6}
\end{array}\right),\,\left(\begin{array}{cc}
-\frac{1}{6} & 0\\
0 & \frac{1}{6}
\end{array}\right),\,\left(\begin{array}{cc}
-\frac{1}{2} & 0\\
0 & \frac{1}{2}
\end{array}\right)$ & $\ket{\textrm{W}^{(3)}}\otimes\ket{1}$ & $\frac{1}{3}$\\
\hline 
$\left(\begin{array}{cc}
-\frac{1}{2} & 0\\
0 & \frac{1}{2}
\end{array}\right),\,\left(\begin{array}{cc}
0 & 0\\
0 & 0
\end{array}\right),\,\left(\begin{array}{cc}
0 & 0\\
0 & 0
\end{array}\right),\,\left(\begin{array}{cc}
0 & 0\\
0 & 0
\end{array}\right)$ &  BiSep & $\frac{3}{8}$\\
\hline 
$\left(\begin{array}{cc}
-\frac{1}{4} & 0\\
0 & \frac{1}{4}
\end{array}\right),\,\left(\begin{array}{cc}
-\frac{1}{4} & 0\\
0 & \frac{1}{4}
\end{array}\right),\,\left(\begin{array}{cc}
-\frac{1}{4} & 0\\
0 & \frac{1}{4}
\end{array}\right),\,\left(\begin{array}{cc}
-\frac{1}{4} & 0\\
0 & \frac{1}{4}
\end{array}\right)$ &  W & $\frac{3}{8}$\\
\hline 
$\left(\begin{array}{cc}
-\frac{1}{10} & 0\\
0 & \frac{1}{10}
\end{array}\right),\,\left(\begin{array}{cc}
-\frac{1}{10} & 0\\
0 & \frac{1}{10}
\end{array}\right),\,\left(\begin{array}{cc}
-\frac{1}{5} & 0\\
0 & \frac{1}{5}
\end{array}\right),\,\left(\begin{array}{cc}
-\frac{1}{5} & 0\\
0 & \frac{1}{5}
\end{array}\right)$ &  $\Phi_3$ & $\frac{9}{20}$\\
\hline 
$\left(\begin{array}{cc}
-\frac{1}{6} & 0\\
0 & \frac{1}{6}
\end{array}\right),\,\left(\begin{array}{cc}
-\frac{1}{6} & 0\\
0 & \frac{1}{6}
\end{array}\right),\,\left(\begin{array}{cc}
-\frac{1}{6} & 0\\
0 & \frac{1}{6}
\end{array}\right),\,\left(\begin{array}{cc}
0 & 0\\
0 & 0
\end{array}\right)$ & $\Phi_2$ & $\frac{11}{24}$\\
\hline 
$\left(\begin{array}{cc}
-\frac{1}{14} & 0\\
0 & \frac{1}{14}
\end{array}\right),\,\left(\begin{array}{cc}
-\frac{1}{14} & 0\\
0 & \frac{1}{14}
\end{array}\right),\,\left(\begin{array}{cc}
-\frac{1}{14} & 0\\
0 & \frac{1}{14}
\end{array}\right),\,\left(\begin{array}{cc}
-\frac{1}{7} & 0\\
0 & \frac{1}{7}
\end{array}\right)$ &  $\Phi_1$ & $\frac{27}{56}$\\
\hline 
$\left(\begin{array}{cc}
0 & 0\\
0 & 0
\end{array}\right),\,\left(\begin{array}{cc}
0 & 0\\
0 & 0
\end{array}\right),\,\left(\begin{array}{cc}
0 & 0\\
0 & 0
\end{array}\right),\,\left(\begin{array}{cc}
0 & 0\\
0 & 0
\end{array}\right)$ &  GHZ & $\frac{1}{2}$\\
\hline 
\end{tabular}
\caption{Four qubits critical one-qubit matrices. The listed states are $\ket{\textrm{TriSep}}=\frac{1}{\sqrt{2}}\ket{11}\otimes (\ket{00}+\ket{11})$, $\ket{\textrm{BiSep}}=\frac{1}{\sqrt{2}}\ket{1}\otimes (\ket{000}+\ket{111})$, $\ket{\textrm{W}}=\frac{1}{2}(\ket{1110}+\ket{1101}+\ket{1011}+\ket{0111})$, $\ket{\Phi_3}=\sqrt{\frac{3}{10}}(\ket{1101}+\ket{1110})+\sqrt{\frac{2}{5}}\ket{0011}$,  $\ket{\Phi_2}=\frac{1}{2\sqrt{3}}(\ket{1011}+\ket{1110})-\frac{1}{2}(\ket{0101}+\ket{0011})+\frac{1}{\sqrt{3}}\ket{0110}$, $\ket{\Phi_1}=\sqrt{\frac{3}{14}}(\ket{0011}+\ket{0101}+\ket{1001})+\sqrt{\frac{5}{14}}\ket{1110}$, $\ket{\textrm{GHZ}}=\frac{1}{\sqrt{2}}(\ket{0000}+\ket{1111})$.}
\label{tab:4qubit_old}
\end{table}

\noindent Finally, note that the considered problem has a permutation symmetry.
More precisely, if 
\begin{eqnarray*}
\alpha=\{\alpha_{1},\alpha_{2},\ldots,\alpha_{L}\}\in\Psi(\mathbb{P}(\mathcal{H})),
\end{eqnarray*}
$\alpha_{i}=\mathrm{diag}(-\lambda_{i},\lambda_{i})$ is critical
then any permutation of its $\lambda_{i}$s gives another critical
$\alpha^{\prime}\in\Psi(\mathbb{P}(\mathcal{H}))$. The main problem
with finding critical points using the method outlined above steams
form the need of considering all possible degeneracies in the spectrum
of $\alpha$. In the next sections we described the improved algorithm which is easier
to handle as it preselects a finite set of $\alpha$'s which have chance to be critical points.   

\section{Abelian {\it vs} nonabelian\label{crit-mu}}

The maps $\mu$ and $\Psi$ from section \ref{sec:The-linear-entropy} have a nice geometric
interpretation (see \cite{SHK} for detailed discussion). The complex projective space $\mathbb{P}(\mathcal{H})$ is a symplectic manifold. The map $\mu$ is the momentum map for the symplectic action of $K$ on $\mathbb{P}(\mathcal{H})$. The image $\Psi(\mathbb{P}(\mathcal{H}))$ is called the momentum polytope and its convexity is a direct consequence of the celebrated convexity property of the momentum map \cite{K84}. Using the $K$-invariant scalar product of matrices given by $(A,B)=\Tr(AB^\dagger)$, one can calculate the squared norm of $\mu$
\begin{equation*}
||\mu(\phi)||^2=\sum_{i=1}^L \Tr\left[\left(\rho_i(\phi)-\frac{1}{2}\mathbbm{1}\right)^2\right]=-\frac{L}{2}+\sum_{i=1}^L \Tr\left[\left(\rho_i(\phi)\right)^2\right].
\end{equation*}
Comparing this result with formula (\ref{eq:linear-entropy}), one can see that the linear entropy is in fact (up to some irrelevant constants) given by $||\mu||^2$ and therefore critical sets of $\|\mu\|^2$ and $E(\cdot)$ coincide. The purpose of the following is to summarize results of \cite{K84-1} where critical points of $||\mu||^2$ were studied for an arbitrary compact semisimple $K$. In particular we show how the problem gets simplified if one considers the action of maximal torus of $K$.

Suppose that a compact semisimple Lie group $K$ acts on a compact
symplectic manifold $M$. Let $\mu:M\rightarrow\mathfrak{k}$
be the momentum map for this action, where $\mathfrak{k}$ is the Lie algebra of $K$. To simplify the notation, the action of $g\in K$ on $x\in M$ will be denote by $g.x$. As shown in \cite{K84-1} (see also \cite{SOK12, SOK14}), the function
$f:=||\mu||^{2}:M\rightarrow\mathbb{R}$ is $K$-invariant, i.e. $||\mu(k.x)||^{2}=||\mu(x)||^{2},\ k\in K$ and
$x$ is a critical point of $f$ iff 
\begin{equation}\label{crit}
\widehat{\mu(x)}_{x}:=\frac{d}{dt}\big|_{t=0}e^{\mu(x)t}.x=0,
\end{equation}
that is, when the fundamental vector field generated by $\mu(x)\in \mathfrak{k}$ and evaluated at point
$x$ vanishes. 

Let $T$ be a maximal torus of $K$ and $\mathfrak{t}$ its Lie algebra. Consider the action of $T$ on $M$ and a corresponding momentum map $\mu_{T}:M\rightarrow \mathfrak{t}$. As $T$ is not semisimple  $\mu_T$ is not unique. One can, however, chose $\mu_T$ to be the composition $\mu_{T}:M\rightarrow\mathfrak{k}\rightarrow\mathfrak{t}$ of the unique momentum map for $K$-action on $M$, $\mu$, with the restriction map $\mathfrak{k}\rightarrow\mathfrak{t}$. The restriction map $\mathfrak{k}\rightarrow\mathfrak{t}$ is an orthogonal projection with respect to the $K$-invariant inner product on $\mathfrak{k}$. Therefore, 
\begin{equation}\label{mu-muT}
\mu(x)=\mu_{T}(x)+\alpha(x),
\end{equation}
where $\alpha(x)\in\mathfrak{t}^{\perp}$. In the following we will consider this $\mu_T$ only. For our choice of $\mu_T$ we can define $f_T:=||\mu_T||^2$. Then $x$ is a critical point of $f_T$ if and only if
\begin{equation}\label{crit-T}
\widehat{\mu_T(x)}_{x}:=\frac{d}{dt}\big|_{t=0}e^{\mu_T(x)t}.x=0,
\end{equation}
that is, when the fundamental vector field generated by $\mu_T(x)\in \mathfrak{t}$ and evaluated at point
$x$ vanishes. In the following subsection we give complete characterization of critical points of $f_T$. Then in subsection \ref{sec:crit-mu} we explain how the knowledge of critical points of $f_T$ simplifies the problem of finding critical points of $f$.
\subsection{Critical points of $\|\mu_T\|^2$}
We start by invoking the Atiyah theorem which gives characterization of $\mu_T(M)$ \cite{atiyah}.
\begin{theorem}\label{Atiyah-thm}
Let $M$ be a symplectic manifold acted upon in a symplectic way by a torus $T$. Let $M_T\subset M$ be the set of $T$-fixed points, i.e. points for which $\forall t\in T\,t.x=x$. The image $\mathbb{A}:=\mu_T(M_T)$ consist of the finite number of points, called weights, and the set $\mu_T(M)=\mathrm{conv}(\mu_T(M_T))$, i.e.  is the convex hull of a finite set of points in $\mathfrak{t}$ that are the image under $\mu_T$ of the fixed points of $T$-action. 
\end{theorem}
Note that by (\ref{crit-T}) one can equivalently say that $x\in M$ is a critical point of $f_T$ iff it is a fixed point of the action of $T_\beta:=\{e^{\beta t}:t\in \mathbb{R}\}\subset T$, where $\beta = \mu_T(x)$.  Following \cite{K84-1} we define:
\begin{definition} For any $\beta\in\mathfrak{t}$, let $Z_{\beta}$
be the union of those connected components of the $T_\beta$-fixed points 
for which  
\begin{equation}\label{perpen}
\langle\mu_{T}(x),\beta\rangle=\langle\beta,\beta\rangle,\,\, x\in Z_\beta.
\end{equation}
\end{definition} 
One can write (\ref{perpen}) as 
\begin{equation}
\langle\mu_{T}(x)-\beta,\beta\rangle=0,\,\, x\in Z_\beta.
\end{equation}
Therefore, all points that belong to $Z_{\beta}$ are fixed by the action of $T_{\beta}$
and mapped by $\mu_{T}$ to the hyperplane in $\mathfrak{t}$
containing $\beta$ and perpendicular to the line from $\beta$ to
the origin. One can easily see that $Z_{\beta}$ is invariant under action of $T$. Indeed, if $x\in Z_\beta$ and $t\in T$ then for any $\tilde{t}\in T_\beta$ we have $\tilde{t}t.x=t\tilde{t}.x=tx$, i.e. $t.x$ is fixed point of $T_\beta$. Moreover 
\begin{equation}
\fl\langle\mu_{T}(t.x)-\beta,\beta\rangle=\langle \mathrm{Ad}_t\mu_T(x)-\beta,\beta\rangle=\langle\mu_{T}(x)-\beta,\beta\rangle=0,\,\, x\in Z_\beta,
\end{equation}
as $T$ commutes with  $\mathfrak{t}$. Actually, $Z_{\beta}$ is a $T$-invariant symplectic submanifold of
$M$ \cite{K84-1}.  Thus by theorem \ref{Atiyah-thm} the image $\mu_T(Z_\beta)$  is the convex hull of the image under
$\mu_{T}$ of the fixed point set of $T$-action on $Z_{\beta}$, i.e. is a subset of the finite set $\mathbb{A}$. Summing up we have \cite{K84-1}:
\begin{theorem} \label{critical_lemma} Let $\beta=\mu_T(x)$ for some $x\in M$. Then $x$ is a critical point of $f_T$  if and only if $x\in Z_\beta$. If this is the case then $\beta$ is the closest point to 0 of the convex hull of some nonempty subset of the set $\mathbb{A}$.
\end{theorem}
We will call $\beta\in \mathfrak{t}$  the {\it minimal combination of weights} iff it is a closest point to the origin of the convex hull in $\mathfrak{t}$ of some nonempty subset of the set of weights $\mathbb{A}$.

\subsection{Critical points of $\|\mu\|^2$}\label{sec:crit-mu}

Let $\mathfrak{t}_+$ be a positive Weyl chamber in $\mathfrak{t}$. $K$-orbits have the following property.
\begin{lemma} Let $K.y$ be an orbit of the $K$-action through  $y\in M$. Then there is $x\in K.y$ such that $\mu(x)\in \mathfrak{t}_{+}$.
\end{lemma}\label{chamber}
\begin{proof}
By the equivariance property of the momentum map, that is $\mu(g.x)=\mathrm{Ad}_g\mu(x)$, where $\mathrm{Ad}_g$ stands for the adjoint action of $g\in K$ on $\mathfrak{k}$, a K-orbit in $M$ is mapped onto an adjoint orbit in $\mathfrak{k}$. But every adjoint orbit intersects $\mathfrak{t}_{+}$ in exactly one point \cite{DK00}, which in turn means that there exists $x\in K.y$ such that $\mu(x)\in\mathfrak{t}_{+}$.
\end{proof}
Recall that critical points of $\|\mu\|^2$ are given by some of $K$-orbits in $M$. By lemma \ref{chamber} every $K$-orbit contains a point such that $\mu(x)\in \mathfrak{t}_{+}$. Combining these two, in order to find critical points of $\|\mu\|^2$ it is enough to investigate points $x\in M$ such that $\mu(x)\in\mathfrak{k}_{+}\subset \mathfrak{t}$ (the remaining critical points are obtained by the action of $K$). Let $M(\mathfrak{t}_{+})=\{x\in M:\,\mu(x)\in \mathfrak{t}_{+}\}$. Note that by (\ref{mu-muT}) for $x\in M(\mathfrak{t}_{+})$ we have $\mu(x)=\mu_T(x)$. By (\ref{crit}) and (\ref{crit-T}) this means that $x\in M(\mathfrak{t}_{+})$ is a critical point of $\|\mu\|^2$ if and only if it is a critical point of $\|\mu_T\|^2$.  But we already know how to check if $x$ is a critical point of $\mu_T$. By theorem \ref{critical_lemma} it is the case if and only if $x\in Z_\beta$, where $\beta=\mu(x)$ is a minimal combination of weights. Let us denote by $\mathcal{B}$ the set of all minimal combinations of weights that belong to $\mathfrak{t}_+$. For $\beta\in \mathcal{B}$ we define $C_\beta=K.\big(Z_\beta\cap\mu^{-1}(\beta)\big )$. Note that $\mu(C_\beta)$ is exactly the adjoint orbit of $K$ through $\beta$. The following is now clear:
\begin{theorem}
The critical set of $|\mu\|^2$ is a disjoint union of sets $C_\beta$, where $\beta$ runs over the finite set $\mathcal{B}$. 
\end{theorem}
Note that for some $\beta\in\mathcal{B}$ the corresponding set $C_\beta$ might  be empty. After deletion of theses pathological $\beta$'s we arrive with set $\tilde{\mathcal{B}}$. In the next section we describe properties of $\tilde{\mathcal{B}}$ for $L$-qubit system. Recall that in this setting, points of $\tilde{\mathcal{B}}$ are given by spectra of the reduced one-qubit density matrices corresponding to the critical points of the linear entropy $E(\cdot)$.

\section{Critical points of the linear entropy for $L$ qubits}\label{crit-Lq}

Let $\mathcal{H}$ be the Hilbert space of $L$ distinguishable qubits, $\mathcal{H}=(\mathbb{C}^{2})^{\otimes L}$.  Let $\{\ket{0},\ket{1}\}$ be a basis in $\mathbb{C}^2$ and $\{\ket{i_1}\otimes\ket{i_2}\otimes \ldots\otimes\ket{i_L}:i_k\in\{0,1\}\}$ the corresponding basis of $\mathcal{H}$. The complex projective space  $M=\mathbb{P}(\mathcal{H})$ is a symplectic manifold. We consider the action of the local unitary operations, $K=SU(2)^{\times L}$ on $M$.
In this setting \cite{SHK}, the momentum map $\mu:M\rightarrow\mathfrak{k}$
assigns to an $L$-qubit state $[\psi]\in M$ the collection of its one-qubit reduced density matrices
\begin{equation}
\mu([\psi])=i\Big{[}\rho_{1}([\psi])-\frac{1}{2}\mathbbm{1},\rho_{2}([\psi])-\frac{1}{2}\mathbbm{1},\dots,\rho_{L}([\psi])-\frac{1}{2}\mathbbm{1}\Big{]}.
\end{equation}
The maximal torus $T\subset K$ consists of diagonal matrices of the
form 
\begin{equation}
\fl T=\Bigg{\{}\Bigg{[}\pmatrix{e^{-i\phi_{1}} & 0\cr
0 & e^{i\phi_{1}}
},\pmatrix{e^{-i\phi_{2}} & 0\cr
0 & e^{i\phi_{2}}
},\dots,\pmatrix{e^{-i\phi_{L}} & 0\cr
0 & e^{i\phi_{L}}
}\Bigg{]}:\, \phi_k\in \mathbb{R}\Bigg{\}}.
\end{equation}
The momentum map $\mu_{T}:M\rightarrow\mathfrak{t}$ assigns to a state
$[\psi]\in M$ a collection of $2\times2$ diagonal matrices whose
elements are given by the diagonal elements of $\rho_{i}([\psi])$'s.
\begin{eqnarray}
\mu_{T}([\psi])
=i\Bigg{[}\pmatrix{-\alpha_{1} & 0\cr
0 & \alpha_{1}
},\pmatrix{-\alpha_{2} & 0\cr
0 & \alpha_{2}
},\dots,\pmatrix{-\alpha_{L} & 0\cr
0 & \alpha_{L}
}\Bigg{]},\\  \alpha_{i}=\mathrm{Tr}\big{\{}\rho_{i}([\psi])\ \kb{i}{i}\big{\}}\in\mathbb{R},\nonumber 
\end{eqnarray}
Therefore, $\mu_{T}(M)$ can be identified with a subset
of $\mathbb{R}^{L}$. The following lemma gives the description
of this subset.

\begin{lemma} \label{qubits_image} For $L$ 
qubits $\mu_T(M)=\mathrm{conv}\mathbb{A}_L$ where 
\begin{equation}
\mathbb{A}_L=\{\overline{\alpha}\in\mathbb{R}^{L}:\,\overline{\alpha}=(\alpha_1,\ldots,\alpha_L),\,\alpha_{i}=\pm1/2\},
\end{equation}
i.e. is a hypercube spanned on $2^L$ vertices. 
\end{lemma}

\begin{proof}By theorem \ref{Atiyah-thm}, $\mu_T(M)$ is a convex hull of the image of $T$-fixed points. It is easy to see that only states $\ket{i_1}\otimes\ldots\otimes\ket{i_L}$, where $i_k\in\{0,1\}$ are $T$-fixed points. Applying $\mu_T$ to them the result follows.  \end{proof}

Therefore, the set $\mathcal{B}_L$ for $L$ qubits is given by the closest to zero
points of convex hulls of the hypercube's vertices. Note that it is
sufficient to consider $2,3,\ldots,L$-vertex subsets of it. This
follows from the fact that every $l$-dimensional convex polytope
can be divided into $l$-dimensional simplices that have disjoint interiors. The closest point to the origin of such a polytope belongs to its boundary. Therefore it belongs to the boundary of one of the $l$-dimensional simplices and as such can be written as a linear combination of $l$ vertices. Moreover it is also the closest point to the origin of the convex span of these $l$ vertices. In our case $l\leq L$ and the result follows.  

The above procedure of finding minimal combinations of weights for $L=3$ is presented in figure \ref{fig:3qubit}. Points $\mathrm{v}_{Bi}$ lie in the centers of hypercube's faces, point $\mathrm{v}_{GHZ}$ lies in the center of the whole hypercube and $\mu(\phi_W)$ is the closest to zero point of triangle marked with dotted lines.
\begin{figure}[h]
\centering
\includegraphics[width=.6\textwidth]{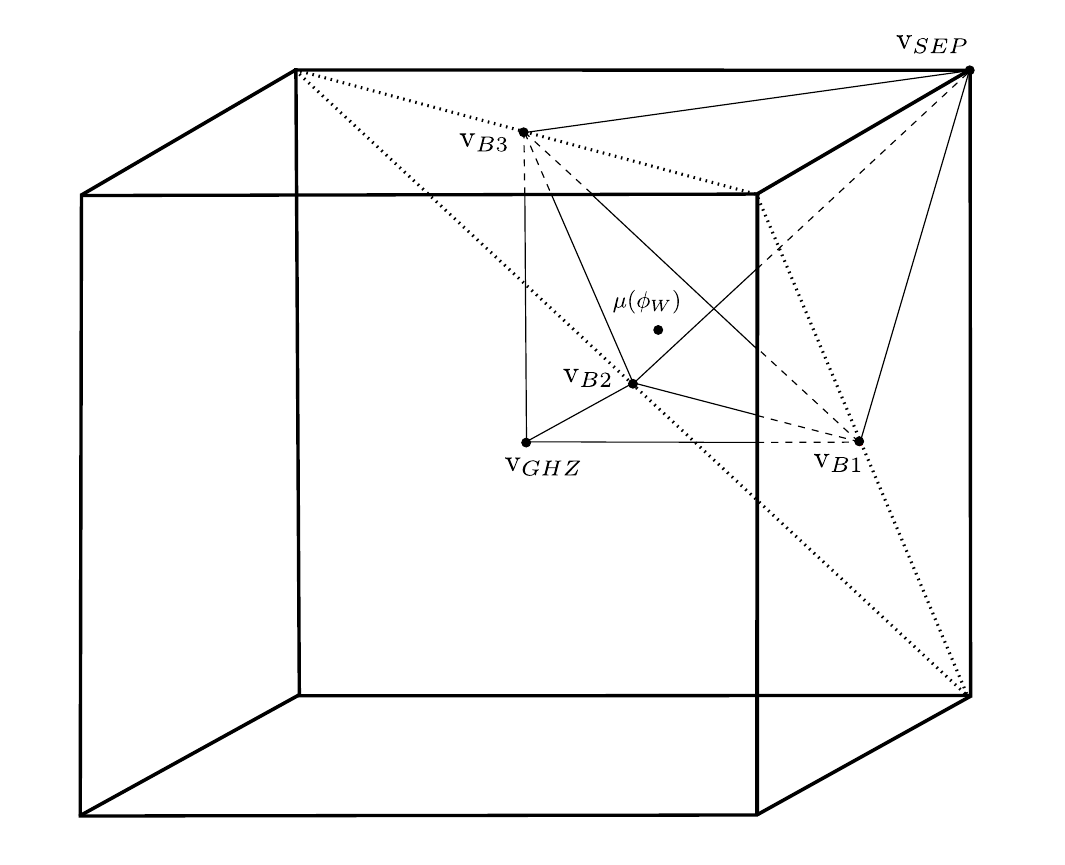}
\caption{The result of the procedure for finding minimal weight combinations for 3 qubits. Point $\mathrm{v}_{GHZ}$ is the image of $\ket{GHZ}=\frac{1}{\sqrt{2}}\left(\ket{000}+\ket{111}\right)$, points $\mathrm{v}_{Bi}$ correspond to biseparable states and state $\ket{\phi_W}=\frac{1}{\sqrt{3}}\left(\ket{110}+\ket{101}+\ket{011}\right)$. }
\label{fig:3qubit}
\end{figure}

Let us next consider the convex hull of the set of vertices $\overline{\alpha}^{(1)},\overline{\alpha}^{(2)},\dots,\overline{\alpha}^{(k)}$.
Any point from the convex hull satisfies 
\begin{equation}
\overline{\beta}=a_{1}\overline{\alpha}^{(1)}+a_{2}\overline{\alpha}^{(2)}+\dots+a_{k}\overline{\alpha}^{(k)},\ \ \sum_{i}a_{i}=1,\ \ a_{i}\in\mathbb{R}_{+}.
\end{equation}
Our goal is to minimize the norm of the vector $\overline{\beta}$ and this way find  $\overline{\beta}_{c}\in \mathcal{B}_L$ such that the following 
\begin{eqnarray}
||\overline{\beta}_{c}||^{2}=\min_{a_{1},\dots,a_{k}}||a_{1}\overline{\alpha}^{(1)}+a_{2}\overline{\alpha}^{(2)}+\dots+a_{k}\overline{\alpha}^{(k)}||^{2}=\\=\min_{a_{1},\dots,a_{k}}\sum_{i=1}^{L}\Big{(}a_{1}\alpha_{i}^{(1)}+\dots+a_{k}\alpha_{i}^{(k)}\Big{)}^{2},\label{minimization_num1}\\
\sum_{i}a_{i}=1,\ \ a_{i}\in\mathbb{R}_{+}.\label{constraints}
\end{eqnarray}
is satisfied. We apply the method
of Lagrange multipliers. To this end, let us consider the Lagrangian
\begin{equation}
\Lambda(a_{1},\dots,a_{k},\lambda):=\sum_{i=1}^{L}\Big{(}a_{1}\alpha_{i}^{(1)}+\dots+a_{k}\alpha_{i}^{(k)}\Big{)}^{2}+\lambda\Big{(}\sum_{i=1}^{k}a_{i}-1\Big{)}\label{lagrangian}
\end{equation}
The conditions for the set of parameters $a_{1},\dots,a_{k}$ to minimize
the norm of the vector $\overline{\beta}$ read 
\begin{equation}
\forall_{l}\ \ \frac{\partial\Lambda}{\partial a_{l}}=0.
\end{equation}
By differentiation of equation (\ref{lagrangian}) one obtains the set
of $L+1$ linear equations for parameters $a_{i}$ and $\lambda$
\begin{equation}
\forall_{l}\ \ 2\sum_{i=1}^{L}\Big{(}a_{1}\alpha_{i}^{(1)}+\dots+a_{k}\alpha_{i}^{(k)}\Big{)}\alpha_{i}^{(l)}+\lambda=0.\label{lin_eqs}
\end{equation}
Using the fact that $\beta=\sum_{i}a_{i}\overline{\alpha}^{(i)}$,
one can rewrite equation (\ref{lin_eqs}) as 
\begin{equation}
\forall_{l}\ \ \overline{\beta}_{c}\cdotp\overline{\alpha}^{(l)}=-\frac{\lambda}{2}.\label{lin1}
\end{equation}
Equation (\ref{lin1}) has a nice geometrical interpretation. As
all vectors $\overline{\alpha}^{(i)}$ are of the same length, 
(\ref{lin1})
means that cosines of the angles between every vector $\overline{\alpha}^{(i)}$
and vector $\overline{\beta}_{c}$ are the same. Equivalently 
\begin{equation}
||\overline{\beta}-\overline{\alpha}^{(1)}||=||\overline{\beta}-\overline{\alpha}^{(2)}||=\dots=||\overline{\beta}-\overline{\alpha}^{(k)}||.
\end{equation}
Equation (\ref{lin_eqs}) can be written as: 
\begin{equation}
\fl 2\sum_{i=1}^{L}\alpha_{i}^{(l)}\sum_{j=1}^{k}a_{j}\alpha_{i}^{(j)}+\lambda=2\sum_{j=1}^{k}a_{j}\Big{(}\sum_{i=1}^{L}\alpha_{i}^{(l)}\alpha_{i}^{(j)}\Big{)}+\lambda=2\sum_{j=1}^{k}a_{j}\overline{\alpha}^{(l)}\cdotp\overline{\alpha}^{(j)}+\lambda=0.\label{lin2}
\end{equation}
Therefore set of equations (\ref{lin_eqs}) has a matrix form: 
\begin{equation}\label{main_eq}
\pmatrix{2\overline{\alpha}^{(1)}\cdotp\overline{\alpha}^{(1)} & 2\overline{\alpha}^{(1)}\cdotp\overline{\alpha}^{(2)} & \dots & 2\overline{\alpha}^{(1)}\cdotp\overline{\alpha}^{(k)} & 1 \cr
2\overline{\alpha}^{(2)}\cdotp\overline{\alpha}^{(1)} & 2\overline{\alpha}^{(2)}\cdotp\overline{\alpha}^{(2)} & \dots & 2\overline{\alpha}^{(2)}\cdotp\overline{\alpha}^{(k)} & 1 \cr
\vdots &  & \ddots &  & \vdots \cr
2\overline{\alpha}^{(k)}\cdotp\overline{\alpha}^{(1)} & 2\overline{\alpha}^{(k)}\cdotp\overline{\alpha}^{(2)} & \dots & 2\overline{\alpha}^{(k)}\cdotp\overline{\alpha}^{(k)} & 1 \cr
1 & 1 & \dots & 1 & 0
}\pmatrix{a_{1} \cr
a_{2} \cr
\vdots \cr
a_{k} \cr
\lambda
}=\pmatrix{0 \cr
0 \cr
\vdots \cr
0 \cr
1
}.
\end{equation}
The matrix is symmetric ($\overline{\alpha}^{(i)}\cdotp\overline{\alpha}^{(j)}=\overline{\alpha}^{(j)}\cdotp\overline{\alpha}^{(i)}$)
and the diagonal entries are $2\overline{\alpha}^{(i)}\cdotp\overline{\alpha}^{(i)}=2L(\frac{1}{2})^{2}=\frac{L}{2}$.
\begin{remark}
Because all entries of matrix (\ref{main_eq}) are rational, it follows from Cartan's theorem that solutions $a_1,\dots,a_k,\lambda$ are also rational. On the other hand, from equation (\ref{lin1}) one has that $||\overline\beta_c||^2=-\frac{\lambda}{2}$. Note that the critical value of the entropy is given by $E_c=\frac{1}{2}-\frac{2}{L}||\overline\beta_c||^2$ and therefore it is a rational number.
\end{remark}

The following lemma explains the relation between critical points for $L$ and $L+1$ qubits. In particular it shows that some of the critical points can be found in an iterative way.
\begin{lemma} Assume $\overline{\beta}_{c}=\sum_{i=1}^{k}a_{i}^{c}\overline{\alpha}^{(i)}\in\mathcal{B}_L$, i.e. $\overline{\beta}_{c}$
is the closest to zero point of the set $\mathrm{conv}\{\overline{\alpha}^{(1)},\dots,\overline{\alpha}^{(k)}\}$ for $L$ qubits. Then 
\begin{enumerate}
\item $\overline{\beta}_{c}'=(\overline{\beta}_{c};1/2)=(\beta_{c,1},\beta_{c,2},\dots,\beta_{c,L},1/2)$
and 
\item$\overline{\beta}_{c}''=(\overline{\beta}_{c};0)=(\beta_{c,1},\beta_{c,2},\dots,\beta_{c,L},0)$, 
\end{enumerate}
belong to $\mathcal{B}_{L+1}$ for  $L+1$ qubits. \label{lemma:iter} 
\end{lemma}

\begin{proof} 
To prove the above statements one needs to show that points $\overline{\beta}_{c}'$ and $\overline{\beta}_{c}''$ are closest points to zero of the convex hulls of some subsets of vertices of $\mathbb{A}$.
\begin{enumerate}
\item Consider the set $\mathrm{conv}\{\overline{\alpha}^{'(1)},\dots,\overline{\alpha}^{'(k)}\}$,
where $\overline{\alpha}^{'(i)}=(\alpha_{1}^{(i)},\dots,\alpha_{L}^{(i)},1/2)$.
The squared norm of a vector from this set is 
\begin{eqnarray}
||\overline{\beta}'||^{2}=||a_{1}\overline{\alpha}^{'(1)}+a_{2}\overline{\alpha}^{'(2)}+\dots+a_{k}\overline{\alpha}^{'(k)}|=\\=\sum_{i=1}^{L}\Big{(}a_{1}\alpha_{i}^{(1)}+\dots+a_{k}\alpha_{i}^{(k)}\Big{)}^{2}+\frac{1}{4}\sum_{i=1}^{k}a_{i}
\end{eqnarray}
Note that $\sum_{i=1}^{k}a_{i}=1$, and hence $\|\beta'\|=\|\beta\|+\frac{1}{4}$, i.e. up to $\frac{1}{4}$ the squared norm of $\overline{\beta}'$ is the same as in the $L$-qubit
case. Therefore, the solution to the problem of finding the closest
to zero point of $\mathrm{conv}\{\overline{\alpha}^{(1)},\dots,\overline{\alpha}^{(k)}\}$
works also in this case. Hence, $\overline{\beta}_{c}'=(\beta_{c,1},\beta_{c,2},\dots,\beta_{c,L},\sum_{i=1}^{k}\frac{1}{2}a_{i}^{c})=(\beta_{c,1},\beta_{c,2},\dots,\beta_{c,L},1/2)$.

\item Assume that critical coefficients, $\{a_{i}^{c}\}_{i=1}^{k}$,
are sorted in a decreasing order, i.e. $a_{1}^{c}\geq a_{2}^{c}\geq\dots\geq a_{k}^{c}$.
Let us now consider the set 
\begin{equation}
\mathrm{conv}\{\overline{\alpha}^{''(1)},\dots,\overline{\alpha}^{''(k)},\overline{\alpha}^{''(k+1)}\},
\end{equation}
where $\overline{\alpha}^{''(1)}=(\alpha_{1}^{(1)},\dots,\alpha_{L}^{(1)},1/2)$,
$\overline{\alpha}^{''(l)}=(\alpha_{1}^{(l)},\dots,\alpha_{L}^{(l)},-1/2)$
for $1<l\leq k$ and $\overline{\alpha}^{''(k+1)}=(\alpha_{1}^{(k)},\dots,\alpha_{L}^{(k)},1/2)$.
By straightforward calculation one checks that the coefficients $a_{l}^{''c}=a_{l}^{c}$
for $1\leq l<k$, $a_{k}^{''c}=a_{1}^{c}+a_{k}^{c}-1/2$ and $a_{k+1}^{''c}=1/2-a_{1}^{c}$
satisfy equation (\ref{main_eq}) with $\lambda''=\lambda=-2||\beta_{c}||^{2}$.
To see this, let us, for example, calculate the first vector element
of the left hand side of equation (\ref{main_eq}) 
\begin{eqnarray*}
\fl\sum_{i=1}^{k+1}2\overline{\alpha}^{''(1)}\cdotp\overline{\alpha}^{''(i)}a_{i}^{''c}+\lambda''=2\sum_{i=1}^{k}\overline{\alpha}^{(1)}\cdotp\overline{\alpha}^{(i)}a_{i}^{c}+2\overline{\alpha}^{(1)}\cdotp\overline{\alpha}^{(k)}\Big{(}a_{1}^{c}-\frac{1}{2}\Big{)}+\frac{1}{2}a_{1}^{c}\\\fl-\frac{1}{2}\Big{(}a_{2}^{c}+\dots+a_{k}^{c}\Big{)}
-\frac{1}{2}\Big{(}a_{1}^{c}-\frac{1}{2}\Big{)}+\Big{(}\frac{1}{2}-a_{1}^{c}\Big{)}\Big{(}2\overline{\alpha}^{(1)}\cdotp\overline{\alpha}^{(k)}+\frac{1}{2}\Big{)}+\lambda=0,
\end{eqnarray*}
where we used the fact that $a_{1}^{c}+\dots+a_{k}^{c}=1$ and $2\sum_{i=1}^{k}\overline{\alpha}^{(1)}\cdotp\overline{\alpha}^{(i)}a_{i}^{c}=-\lambda$.
Finally, by similar calculation one checks that $\overline{\beta}_{c}''=\sum_{i=1}^{k+1}a_{i}^{''c}\overline{\alpha}^{''(i)}=(\beta_{c,1},\beta_{c,2},\dots,\beta_{c,L},0)$. 
\end{enumerate}
\end{proof}
\noindent The remaining critical values and corresponding closest to zero points
can be found by numerical solution of equation (\ref{lin2}). The results
for $3,4$ and $5$ qubits are shown in table \ref{tabelka2}.

Next we give the characterization of the set $Z_{\beta}$ for $L$ qubits. As noted in the proof of lemma \ref{qubits_image} states that are mapped by $\mu$ to the vertices spanning $\mathbb{A}$ are exactly separable states of the form $\ket{i_1}\otimes \ket{i_2}\otimes\ldots\otimes \ket{i_L}$. Moreover, each set $\mu_T(Z_\beta)$ is the convex hull of some subset of vertices belonging to $\mathbb{A}$. One may then expect that $Z_\beta$ is a complex projective space of the vector space spanned by separable states that are mapped by $\mu_T$ on the vertices spanning $\mu_T(Z_\beta)$. The following lemma shows that this is the case.

\begin{lemma}\label{lemma_zbeta} Let $\mathbb{A}_{\beta}=\left\{\overline{\alpha}^{(1)},\overline{\alpha}^{(2)},\dots,\overline{\alpha}^{(k)}\right\}$
be a subset of all vertices of hypercube $\mathbb{H}=\mu_{T}\left(\mathbb{P}\left(\mathcal{H}\right)\right)$ such that $\overline\beta\cdot\overline\alpha^{(i)}=||\overline\beta||^2$ for an arbitrary $\overline\beta\in\mathfrak{t}$.
Let us also denote by $\mathcal{S}\subset\mathcal{H}$ a subspace spanned by separable states that are mapped by $\mu_{T}$ on the vertices $\overline{\alpha}^{(i)}$. Then $Z_{\beta}=\mathbb{P}(\mathcal{S})$.
\end{lemma}

\begin{proof} We start with showing that $\mathbb{P}(\mathcal{S})\subset Z_{\beta}$.
To this end, we introduce the following notation. Vertices from $\mathbb{A}_{\beta}$
are of the form $\overline{\alpha}^{(i)}=\big{(}\sigma_{i_{1}}\frac{1}{2},\sigma_{i_{2}}\frac{1}{2},\dots,\sigma_{i_{L}}\frac{1}{2}\big{)}$,
where $\sigma_{i_{l}}=(-1)^{(1-i_{l})}$ and $i_{l}\in\{0,1\}$. The separable state from $\mathcal{S}$ that is mapped by $\mu_{T}$ on
$\overline{\alpha}^{(i)}$ is the state $\ket{i_{1}}\otimes \ket{i_{2}}\otimes\ldots\otimes \ket{i_{L}}$, i.e.
\begin{equation}
\mu_T(\ket{i_1}\otimes\ldots\otimes\ket{i_L})=\big{(}\sigma_{i_{1}}\frac{1}{2},\sigma_{i_{2}}\frac{1}{2},\dots,\sigma_{i_{L}}\frac{1}{2}\big{)}=\overline{\alpha}^{(i)}.
\end{equation}
An element of the group $T_{\beta}$ can be represented as a matrix acting on $\mathcal{H}$ 
\begin{equation}
\fl k_{t}=\pmatrix{e^{-i\beta_{1}t} & 0 \cr
0 & e^{i\beta_{1}t}
}\otimes\pmatrix{e^{-i\beta_{2}t} & 0\cr
0 & e^{i\beta_{2}t}
}\otimes\dots\otimes\pmatrix{e^{-i\beta_{L}t} & 0\cr
0 & e^{i\beta_{L}t}
},\ t\in\mathbb{R}.
\end{equation}
We will next show that all separable states from $\mathcal{S}$ belong
to the same eigenspace of $k_{t}$ for any $t\in \mathbb{R}$. Matrix $k_{t}$ acts on an
arbitrary separable state from $\mathcal{S}$ in the following way
\begin{equation}
\fl k_{t}.\ket{i_{1}}\otimes \ket{i_{2}}\otimes\dots\otimes \ket{i_{L}}=\exp\left(it\sum_{l=1}^{L}\sigma_{i_{l}}\beta_{l}\right)\ket{i_{1}}\otimes \ket{i_{2}}\otimes\dots\otimes \ket{i_{L}}.\label{eq:exp_coefficient-1}
\end{equation}
Note that sum appearing in the exponent in (\ref{eq:exp_coefficient-1})
can be written as 
\begin{equation}
\sum_{l=1}^{L}\sigma_{i_{l}}\beta_{l}=2\overline{\beta}\cdotp\overline{\alpha}^{(i)}=2||\overline{\beta}||^{2},
\end{equation}
which means that it does not depend on $i$. Hence, all separable states from $\mathcal{S}$ belong to the same eigenspace of $k_{t}$ with eigenvalue $\exp(2i||\overline{\beta}||^{2}t)$. As a direct consequence, all states from $\mathbb{P}(\mathcal{S})$ are fixed by the action of $T_{\beta}$. We still need to show that
for any $\ket{\Psi}\in\mathcal{S}$, $\overline{\beta}\cdotp\mu_{T}\left(\ket{\Psi}\right)=||\overline{\beta}||^{2}$.
A state $\ket{\Psi}\in\mathcal{S}$ is of the form 
\begin{equation}
\ket{\Psi}=\sum_{i=1}^{k}c_{i}\ket{i_{1}}\otimes \ket{i_{2}}\otimes\dots\otimes \ket{i_{L}},\label{eq:lem_zbeta_state}
\end{equation}
where $\sum_{i=1}^{k}|c_{i}|^{2}=1.$ The diagonal entries of
the $l$th one-qubit reduced density matrix read 
\[
\rho_{00}^{(l)}=\sum_{i:i_{l}=0}|c_{i}|^{2},\,\,\rho_{11}^{(l)}=\sum_{i:i_{l}=1}|c_{i}|^{2}=1-\rho_{00}^{(l)}.
\]
Therefore, the $l$th component of vector $\overline{\eta}=\mu_{T}(\ket{\Psi})$ is 
\begin{equation}
\fl \eta_{l}=\rho_{11}^{(l)}-\frac{1}{2}=\sum_{i:i_{l}=1}|c_{i}|^{2}-\frac{1}{2}=\sum_{i:i_{l}=1}|c_{i}|^{2}-\frac{1}{2}\sum_{i=1}^{k}|c_{i}|^{2}=\sum_{i=1}^{k}\frac{1}{2}(-1)^{(1-i_{l})}|c_{i}|^{2}.\label{rho_l-1}
\end{equation}
Applying (\ref{rho_l-1}) to each element of vector $\overline{\eta}$
and using the fact that $\overline{\alpha}^{(i)}=\big{(}\sigma_{i_{1}}\frac{1}{2},\sigma_{i_{2}}\frac{1}{2},\dots,\sigma_{i_{L}}\frac{1}{2}\big{)}$,
where $\sigma_{i_{l}}=(-1)^{(1-i_{l})}$, one obtains 
\[
\overline{\eta}=\sum_{i=1}^{k}|c_{i}|^{2}\overline{\alpha}^{(i)}.
\]
 Finally, let us calculate
\[
\overline{\beta}\cdotp\overline{\eta}=\sum_{i=1}^{k}|c_{i}|^{2}\overline{\beta}\cdotp\overline{\alpha}^{(i)}=||\overline{\beta}||^{2}\sum_{i=1}^{k}|c_{i}|^{2}=||\overline{\beta}||^{2},
\]
which finishes the proof that $\mathbb{P}(\mathcal{S})\subset Z_{\beta}$.

The final step to prove equality of those sets is to show that $Z_{\beta}\cap\mathbb{P}(\mathcal{H}\setminus\mathcal{S})=\emptyset$.
Assume there exists $\ket{\Psi}\in Z_{\beta}$ that does not belong
to $\mathbb{P}(\mathcal{S})$. It follows that there exists at least
one additional separable tensor, $\ket{a_{1}}\otimes \ket{a_{2}}\otimes\dots\otimes \ket{a_{L}}\notin\mathcal{S}$
that appears in the decomposition of $\ket{\Psi}$, i.e. 
\begin{eqnarray}
[\Psi]=\alpha \ket{a_{1}}\otimes \ket{a_{2}}\otimes\dots\otimes \ket{a_{L}}+\sum_{i=1}^{k}c_{i}\ket{i_{1}}\otimes \ket{i_{2}}\otimes\dots\otimes \ket{i_{L}}=\\
=\alpha \ket{a_{1}}\otimes \ket{a_{2}}\otimes\dots\otimes \ket{a_{L}}+\Psi_{\mathcal{S}},\,\,\Psi_{\mathcal{S}}\in\mathbb{P}(\mathcal{S}).
\end{eqnarray}
As state $\ket{\Psi}$ is fixed by the action of $T_{\beta}$
and $T_{\beta}$ is diagonal in the basis of separable states,
for arbitrary $t\in \mathbb{R}$ the state $\ket{a_{1}}\otimes \ket{a_{2}}\otimes\dots\otimes \ket{a_{L}}$
belongs to the same eigenspace of $k_{t}\in T_{\beta}$ as
the remaining separable states from $\mathbb{P}(\mathcal{S})$. Applying
(\ref{eq:exp_coefficient-1}) to
$\ket{a_{1}}\otimes \ket{a_{2}}\otimes\dots\otimes \ket{a_{L}}$
one can see that the scalar product of $\overline{\alpha}^{(a)}=\mu_{T}\left(\ket{a_{1}}\otimes \ket{a_{2}}\otimes\dots\otimes \ket{a_{L}}\right)$
and $\overline{\beta}$ must be then equal to $||\overline{\beta}||^{2}$.
This in turn implies that $\overline{\alpha}^{(a)}\in\mathbb{A}_{\beta}$,
which is a contradiction.\end{proof}

The method for finding critical values of $||\mu||^{2}$ described
in this section raises another issue. We know that critical
sets of $||\mu||^{2}$ are the $K$-orbits through sets $Z_{\beta_{c}}\cap\mu^{-1}\left(\overline{\beta_{c}}\right)$,
where $\overline{\beta_{c}}$ is a point from $\mathfrak{t}_{+}$
that is the closest to zero point of a hyperplane spanned by subset of vertices
of the hypercube $\mathbb{H}=\mu_{T}\left(\mathbb{P}\left(\mathcal{H}\right)\right)$.
In order to decide whether $\overline{\beta}_{c}$ is in the image under $\mu$ of some critical set, one has to know that
$Z_{\beta_{c}}\cap\mu^{-1}\left(\overline{\beta_{c}}\right)\neq\emptyset$. In the following, we show that this is the case for almost all critical values $\overline{\beta}_{c}\in\mathfrak{t}_{+}$. Let us first prove a useful lemma.

\begin{lemma} \label{lemma_state}For $\overline{\beta}_{c}\in\mathfrak{t}_{+}$,
a closest to zero point of set $\mathrm{conv}\left\{ \overline{\alpha}^{(1)},\overline{\alpha}^{(2)},\dots,\overline{\alpha}^{(k)}\right\} $,
$\overline{\beta}_{c}=\sum_{i=1}^{k}a_{i}^{c}\overline{\alpha}^{(i)}$,
define state $\ket{\Psi}\in\mathcal{H}$
\begin{equation}
\ket{\Psi}=\sum_{i=1}^{k}\sqrt{a_{i}^{c}}\ket{i_{1}}\otimes \ket{i_{2}}\otimes\dots\otimes \ket{i_{L}},\label{eq:lem_state}
\end{equation}
where $\mu_{T}\left(\ket{i_{1}}\otimes \ket{i_{2}}\otimes\dots\otimes \ket{i_{L}}\right)=\overline{\alpha}^{(i)}$.
Then $\mu_{T}\left(\ket{\Psi}\right)=\overline{\beta}_{c}$.
Morever, if $\beta_{c,m}>0$ for all $m$ then the one-qubit reduced density
matrices of state $\ket{\Psi}$ are diagonal which in particular means that $\mu(\ket{\Psi})=\mu_T(\ket{\Psi})$. \end{lemma}

\begin{proof}Calculation showing that $\mu_{T}(\ket{\Psi})=\overline{\beta}_{c}$
is analogous to the one presented in the proof of lemma \ref{lemma_zbeta} (put $c_{i}=\sqrt{a_{i}^{c}}$
in (\ref{eq:lem_zbeta_state})). In order to show the remaining part assume that the $m$-th one-qubit reduced
density matrix, $\rho_{m}(\ket{\Psi})$, is non-diagonal while all $\beta_{c,i}>0$.
This implies that there exist such $i$, $j$ that $i$-th and $j$-th
basis vectors from (\ref{eq:lem_state}) satisfy $\ket{i_{l}}=\ket{j_{l}}$
for all $l\neq m$ and $\ket{i_{m}}\neq \ket{j_{m}}$. These basis vectors
are mapped by $\mu_{T}$ to vertices $\overline{\alpha}^{(i)}$ and
$\overline{\alpha}^{(j)}$ such that $\alpha_{l}^{(i)}=\alpha_{l}^{(j)}$
for $l\neq m$ and $\alpha_{m}^{(i)}\neq\alpha_{m}^{(j)}$. From the
fact that $\overline{\beta}_{c}$ is the closest to zero point, it
follows that
\begin{equation}
\overline{\alpha}^{(i)}\cdotp\overline{\beta}_{c}=\overline{\alpha}^{(j)}\cdotp\overline{\beta}_{c}.\label{eq:lem_prods}
\end{equation}
As $\alpha_{m}^{(i)}\neq\alpha_{m}^{(j)}$ and the remaining
components of vectors $\overline{\alpha}^{(i)}$ and $\overline{\alpha}^{(j)}$
are the same, (\ref{eq:lem_prods}) can be satisfied only if
$\beta_{c,m}=0$. This is a contradiction. \end{proof}
\noindent The theorem which we promised to prove is the following:

\begin{theorem}\label{z_beta_theorem} Let $\overline{\beta}_{c}\in\mathfrak{t_{+}}$ be
the closest to zero point of the convex hull of some subset of vertices
of the polytope $\mathbb{H}=\mu_{T}\left(\mathbb{P}\left(\mathcal{H}\right)\right)$
for $L$ qubits. Set $Z_{\beta_{c}}\cap\mu^{-1}\left(\overline{\beta_{c}}\right)\subset\mathbb{P}\left(\mathcal{H}\right)$
is nonempty iff $\overline{\beta}_{c}\neq\left(\frac{1}{2},\frac{1}{2},\dots,\frac{1}{2},0\right)$, up to permutations of vector components.\end{theorem}

\begin{proof} The proof consists of two parts. First we consider
the case, where all components of vector $\overline{\beta}_{c}$ are
greater than zero. In the second part we prove the theorem for $\overline{\beta}_{c}$
belonging to the boundary of $\mathfrak{t_{+}}$, i.e. when $|\{k:\,\beta_{c,k}=0\}|\neq0.$ 
\begin{description}

\item [{a}] Let $\left(\overline{\alpha}^{(1)},\overline{\alpha}^{(2)},\dots,\overline{\alpha}^{(k)}\right)$
be the set of vertices, for which the convex hull $\overline{\beta}_{c}=\sum_{i=1}^{k}a_{i}^{c}\overline{\alpha}^{(i)}$
is the closest point to the origin. Moreover, assume that $\beta_{c,m}>0$
for all $m$. One can construct state 
\[
\ket{\Psi}=\sum_{i=1}^{k}\sqrt{a_{i}^{c}}\ket{i_{1}}\otimes \ket{i_{2}}\otimes\dots\otimes \ket{i_{L}},
\]
where $\mu_{T}\left(\ket{i_{1}}\otimes \ket{i_{2}}\otimes\dots\otimes \ket{i_{L}}\right)=\overline{\alpha}^{(i)}$.
From lemma \ref{lemma_zbeta} we know that $\ket{\Psi}\in Z_{\beta_{c}}$
and from lemma \ref{lemma_state} it follows that the one-qubit reduced
density matrices of $\ket{\Psi}$ are diagonal. Thus $\mu\left(\ket{\Psi}\right)=\mu_{T}\left(\ket{\Psi}\right)=\overline{\beta}_{c}$.
Therefore, $\ket{\Psi}\in Z_{\beta_{c}}\cap\mu^{-1}\left(\overline{\beta_{c}}\right)$.

\item [{b}] Now let $\mathbb{A}_{\beta_c}=\left(\overline{\alpha}^{(1)},\overline{\alpha}^{(2)},\dots,\overline{\alpha}^{(k)}\right)$
be the set of all vertices that belong to the hyperplane, for which
$\overline{\beta}_{c}=\sum_{i=1}^{k}a_{i}^{c}\overline{\alpha}^{(i)}$
is the closest point to zero. Here, we consider the case when $M$
out of $L$ components of $\overline{\beta}_{c}$ are nonzero.
By permutation symmetry we can assume that $\beta_{c,l}>0$ for $l\leq M$
and $\beta_{c,l}=0$ for $l>M$. Let us  construct $M$-dimensional vector
$\overline{\tilde{\beta}}_{c}$ such that $\tilde{\beta}_{c,m}=\beta_{c,m}$,
i.e. all components of $\overline{\tilde{\beta}}_{c}$ are greater
than zero. Similarly, construct $\tilde{\mathbb{A}}$, i.e. the
set of $k$ $M$-dimensional vectors, where $\overline{\tilde{\alpha}}_{m}^{(i)}=\overline{\alpha}_{m}^{(i)}$.
Then from the fact that $\beta_{c,l}=0$ for $l>M$ it follows that
$\overline{\tilde{\alpha}}^{(i)}\cdotp\overline{\tilde{\beta}}_{c}=\overline{\alpha}^{(i)}\cdotp\overline{\beta}_{c}=\left\Vert \overline{\tilde{\beta}}_{c}\right\Vert ^{2}$
for all $i$. Therefore, $\overline{\tilde{\beta}}_{c}$ is the closest
to zero point for the convex hull of set $\tilde{\mathbb{A}}_{k}$,
but for $M$ qubits. By the construction of the algorithm
for finding minimal combinations of weights we know that it is possible
to choose a set $\mathbb{B}\subset\tilde{\mathbb{A}}$ of at most
$M$ vertices, for which convex hull $\overline{\tilde{\beta}}_{c}$
is still the closest to zero point. Let us emphasize the fact that $|\mathbb{B}|\leq M$,
which will be very important for our proof. Denote by $\overline{\beta}^{(i)}$
the $i$th vertex from set $\mathbb{B}$. Then $\overline{\tilde{\beta}}_{c}=\sum_{i=1}^{|\mathbb{B}|}b_{i}^{c}\overline{\beta}^{(i)}$
and we can apply the same reasoning as in part $\textbf{a}$ and construct
$M$ qubit state $\ket{\Psi}\in Z_{\tilde{\beta}_{c}}$ such
that $\mu\left(\ket{\Psi}\right)=\mu_{T}\left(\ket{\Psi}\right)=\overline{\tilde{\beta}}_{c}$.
Namely 
\begin{equation}
\ket{\Psi}=\sum_{i=1}^{|\mathbb{B}|}\sqrt{b_{i}^{c}}\ket{i_{1}}\otimes \ket{i_{2}}\otimes\dots\otimes \ket{i_{M}},\label{eq:psi2}
\end{equation}
where $\mu_{T}\left(\ket{i_{1}}\otimes \ket{i_{2}}\otimes\dots\otimes \ket{i_{M}}\right)=\overline{\beta}^{(i)}$.
Our goal will be now to use $\left[\Psi\right]$ in order to construct
a $L$ qubit state, $\ket{\Phi}$, that belongs to $Z_{\beta_{c}}\cap\mu^{-1}\left(\overline{\beta_{c}}\right)$.
Equivalently, we want $\ket{\Phi}$ to be an eigenvector of
matrix
\[ \fl
k_{t}=\pmatrix{e^{-i\beta_{c,1}t} & 0\cr
0 & e^{i\beta_{c,1}t}
}\otimes\dots\otimes\pmatrix{e^{-i\beta_{c,M}t} & 0\cr
0 & e^{i\beta_{c,M}t}
}\otimes\pmatrix{1 & 0\cr
0 & 1
}\otimes\dots\otimes\pmatrix{1 & 0\cr
0 & 1
},
\]
where $t\in \mathbb{R}$ and to satisfy $\mu(\ket{\Phi})=\mu_{T}(\ket{\Phi})=\overline{\beta}_{c}$.
We will search for $\ket{\Phi}$ of the form
\begin{equation}
\ket{\Phi}=\frac{1}{\sqrt{2}}\left(\ket{\Psi}\otimes \ket{0}\otimes\dots\otimes \ket{0}+\ket{\Psi^{\perp}}\otimes \ket{1}\otimes\dots\otimes \ket{1}\right),\label{eq:phi_state}
\end{equation}
where$\ket{\Psi^{\perp}}$ is perpendicular to $\ket{\Psi}$,
belongs to $Z_{\tilde{\beta}_{c}}$ and 
\begin{equation}
\mu(\ket{\Psi^{\perp}})=\mu_{T}(\ket{\Psi^{\perp}})=\overline{\tilde{\beta}}_{c}.
\end{equation}
When $\ket{\Psi^{\perp}}$ satisfies all these conditions,
it is a matter of straightforward calculation to see that the one-qubit
reduced density matrices of $\ket{\Phi}$ are
\begin{equation}
\rho_{l}\left(\ket{\Phi}\right)=\rho_{l}\left(\ket{\Psi}\right)\ \ \mathrm{for}\ i=1,\dots,M\ \ \mathrm{and}\ \rho_{l}\left(\ket{\Phi}\right)=\pmatrix{\frac{1}{2} & 0\cr
0 & \frac{1}{2}
}\  \label{eq:1qrdm_phi}
\end{equation}
for $i=M+1,\dots,L$, i.e. $\mu(\ket{\Phi})=\mu_{T}(\ket{\Phi})=\overline{\beta}_{c}$.
Moreover, $\ket{\Phi}$ is an eigenvector of $k_{t}$, because
both $\ket{\Psi^{\perp}}$ and $\ket{\Psi}$ belong
to $Z_{\tilde{\beta}_{c}}$ and the part of $k_{t}$ acting on the
last $L-M$ qubits is identity.

One way to find desired $\ket{\Psi^{\perp}}$ is to construct
such $k\in T=\mathrm{Stab}\left(\mu_{T}\left(\ket{\Psi}\right)\right)$ that $\ket{k.\Psi}=\ket{\Psi^{\perp}}$.
General form of such $k$ is
\[
k=\pmatrix{e^{-i\phi_{1}} & 0\cr
0 & e^{i\phi_{1}}
}\otimes\pmatrix{e^{-i\phi_{2}} & 0\cr
0 & e^{i\phi_{2}}
}\otimes\dots\otimes\pmatrix{e^{-i\phi_{M}} & 0\cr
0 & e^{i\phi_{M}}
}.
\]
Therefore, the action of $k$ on $\left[\Psi\right]$ gives
\[
\fl k.\ket{\Psi}=\sum_{i=1}^{|\mathbb{B}|}\sqrt{b_{i}^{c}}\left(\prod_{k=1}^{M}e^{i\sigma_{i_{k}}\phi_{k}}\right)\ket{i_{1}}\otimes \ket{i_{2}}\otimes\dots\otimes \ket{i_{M}},\ \mathrm{where}\ \sigma_{i_{k}}=(-1)^{(1-i_{k})}.
\]
The condition $\ket{\Psi}\perp \ket{k.\Psi}$ then reads
\begin{equation}
\sum_{i=1}^{|\mathbb{B}|}b_{i}^{c}\left(\prod_{k=1}^{M}e^{i\sigma_{i_{k}}\phi_{k}}\right)=0.\label{eq:perp_cond}
\end{equation}
Equation (\ref{eq:perp_cond}) can be viewed as a problem of constructing
a polygon with the lengths of the sides given and the angles parametrized
by $\{\phi_{k}\}_{k=1}^{M}$. In other words, we have a set of complex
numbers given by
\[
z_{j}=b_{j}^{c}\exp\left(i\gamma_{j}\right),\ \mathrm{where}\ \gamma_{j}=\sum_{k=1}^{M}\sigma_{j_{k}}\phi_{k}
\]
and we want them to add up to zero (Fig.\ref{fig:polygon}). 
\begin{figure}[H]
\centering \includegraphics[width=0.7\textwidth]{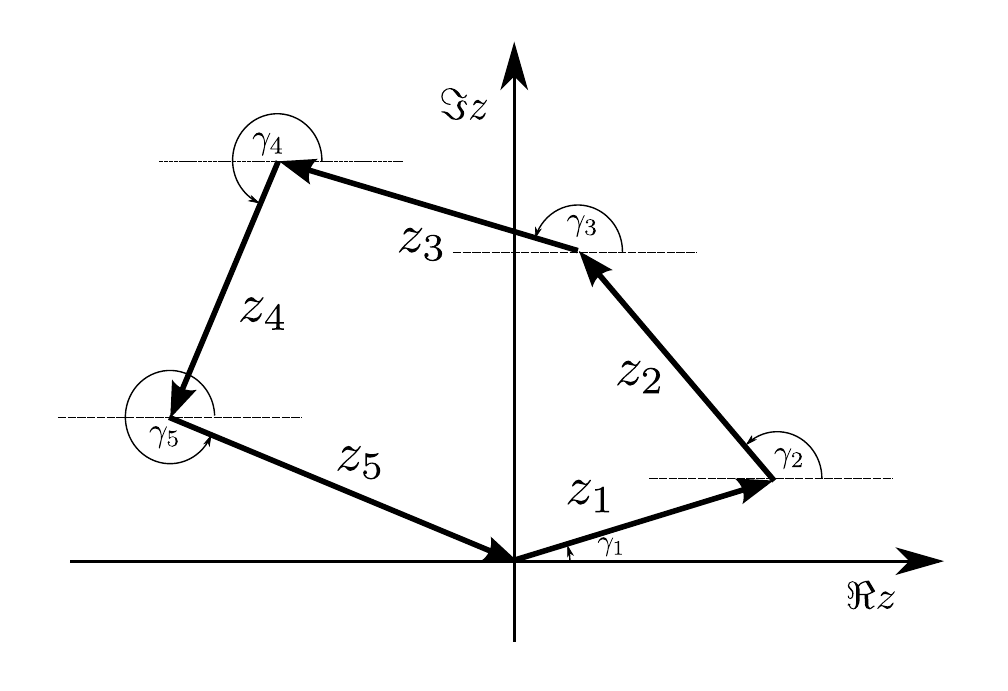} \caption{\label{fig:polygon} An example of the geometrical interpretation
of equation (\ref{eq:perp_cond}).}
\end{figure}
The conditions sufficient to construct such polygon are: 

\begin{enumerate}
\item the length of the longest side is smaller than the sum of the lengths
of the remaining sides, 
\item there are at least 2 sides and at most $M$ sides. 
\end{enumerate}

The second condition means that the number of the angles of the polygon
cannot be greater than the number of variables $\{\phi_{k}\}_{k=1}^{M}$
parameterizing those angles. Note that the number of the sides of the
polygon is equal to the number of the separable states that span state
$\ket{\Psi}$. On the other hand, $\ket{\Psi}$ belongs to the span of $|\mathbb{B}|$
separable states (see (\ref{eq:psi2})), which is less or equal
to $M$. Hence, the second condition is satisfied, provided that $\overline{\beta}_{c}\neq\left(\frac{1}{2},\frac{1}{2},\dots,\frac{1}{2},0\right)$. If $\overline{\beta}_{c}=\left(\frac{1}{2},\frac{1}{2},\dots,\frac{1}{2},0\right)$, then $\ket{\Psi}=\ket{11\dots1}$ and there is only one vector on figure \ref{fig:polygon}, which can't be made equal to zero.  In other words, $\ket{k.\Psi}=\ket{\Psi}$ for all $k\in T$. The following lemma shows that the first condition is satisfied as well.
\end{description}
\end{proof}

\begin{lemma} For a critical state $\ket{\Psi}=\sum_{i}c_{i}\ket{i_{1}}\otimes \ket{i_{2}}\otimes\dots\otimes \ket{i_{M}}$
whose all one-qubit reduced density matrices are not maximally mixed
($\beta_{c,k}>0$ for all $k$), the set $\{|c_{i}|^{2}\}$ satisfies
the polygon inequality, i.e. if $|c_{\mathrm{max}}|^{2}:={\displaystyle \max_{i}|c_{i}|^{2}}$,
then 
\begin{equation}
|c_{\mathrm{max}}|^{2}\leq\sum_{i\neq\mathrm{max}}|c_{i}|^{2}.\label{polygon}
\end{equation}
\end{lemma}

\begin{proof} Let us assume that (\ref{polygon}) is not satisfied,
i.e. that there exists such $c_{\mathrm{max}}$ that 
\begin{equation}
|c_{\mathrm{max}}|^{2}>\sum_{i\neq\mathrm{max}}|c_{i}|^{2}.\label{contr}
\end{equation}
We will show that this leads to a contradiction, i.e. a state for
which (\ref{contr}) is true cannot be critical. First, let us write
matrix elements of each one-qubit density matrix 
\begin{eqnarray}
\rho_{k}^{(00)}\left(\ket{\Psi}\right)=\sum_{i:\, i_{k}=0}|c_{i}|^{2},\label{rho0}\\
\rho_{k}^{(11)}\left(\ket{\Psi}\right)=\sum_{i:\, i_{k}=1}|c_{i}|^{2}.\label{rho1}
\end{eqnarray}
The condition $\beta_{c,k}>0$ is equivalent to 
\begin{equation}
\rho_{k}^{(11)}\left(\ket{\Psi}\right)>\rho_{k}^{(00)}\left(\ket{\Psi}\right).\label{rho}
\end{equation}
By (\ref{contr}), inequality (\ref{rho}) can
be satisfied only if sum (\ref{rho1}) includes term $|c_{\mathrm{max}}|^{2}$
, i.e. 
\begin{equation}
\rho_{k}^{(11)}\left(\ket{\Psi}\right)=|c_{\mathrm{max}}|^{2}+\sum_{i\neq\mathrm{max},\, i:\, i_{k}=1}|c_{i}|^{2}.\label{rho_1p}
\end{equation}
Note that according to equation (\ref{rho_1p}) $|c_{\mathrm{max}}|^{2}$
appears in the expression for the matrix element $\rho_{k}^{(11)}\left(\ket{\Psi}\right)$
of each one-qubit reduced density matrix. Therefore, $c_{\mathrm{max}}$
is the coefficient multiplying the basis vector $\ket{1}\otimes \ket{1}\otimes\dots\otimes \ket{1}$.
This implies that $\ket{\Psi}$ belongs to the eigenspace of $k_{t}$
with eigenvalue $\exp\left(-it\sum_{k=1}^{M}\beta_{c,k}\right)$ (for
calculation, see equation (\ref{eq:exp_coefficient-1}) from lemma \ref{lemma_zbeta}).
On the other hand, one can checks by direct calculation that $\ket{\Psi}$
belongs to the eigenspace of $k_{t}$ with eigenvalue $\exp\left(2it||\overline{\beta_{c}}||^{2}\right)$
(see again the proof of lemma \ref{lemma_zbeta}). This is a contradiction,
because those two eigenvalues are never equal.\end{proof}

\paragraph*{Example of constructing a critical state} Here we describe how to apply the above theorems concerning the construction of the $Z_\beta$ set and finding an exemplary critical state. Let us focus on a $4$ qubit case. Assume that we have found a $3$ qubit critical state $\ket{\phi_W}=\frac{1}{\sqrt{3}}\left(\ket{110}+\ket{101}+\ket{011}\right)$ with $\mu(\phi_W)=\mu_T(\phi_W)=\left(\frac{1}{6},\frac{1}{6},\frac{1}{6}\right)$ (see Table \ref{tab:3qubit_old}). From lemma \ref{lemma:iter} we know that $\overline\beta=\left(\frac{1}{6},\frac{1}{6},\frac{1}{6},0\right)$ is the closest to zero point of the convex hull of some vertices from the $4$-qubit hypercube. The set $\mathbb{A}_\beta$ of vertices contained in the hyperplane perpendicular to $\beta$ is given by
\begin{eqnarray*}
\mathbb{A}_\beta=\Bigg{\{}\left(\frac{1}{2},\frac{1}{2},-\frac{1}{2},\frac{1}{2}\right),\left(\frac{1}{2},-\frac{1}{2},\frac{1}{2},\frac{1}{2}\right),\left(-\frac{1}{2},\frac{1}{2},\frac{1}{2},\frac{1}{2}\right), \\ \left(\frac{1}{2},\frac{1}{2},-\frac{1}{2},-\frac{1}{2}\right),\left(\frac{1}{2},-\frac{1}{2},\frac{1}{2},-\frac{1}{2}\right),\left(-\frac{1}{2},\frac{1}{2},\frac{1}{2},-\frac{1}{2}\right)\Bigg{\}}.
\end{eqnarray*}
Indeed, one can easily check that $\overline\alpha^{(i)}\cdot\overline\beta=||\overline\beta||^2=\frac{1}{12}$ for all $\overline\alpha^{(i)}\in\mathbb{A}_\beta$. From lemma \ref{lemma_zbeta} one has that $Z_\beta=\mathbb{P}(\mathcal{S})$, where
\begin{equation*}
\mathcal{S}=\mathrm{Span}_{\mathbb{C}}\{\ket{1101},\ket{1011},\ket{0111},\ket{1100},\ket{1010},\ket{0110}\}.
\end{equation*}
Let us next construct a critical state $\ket{\Phi}\in\mu^{-1}(\beta)$ following the proof of theorem \ref{z_beta_theorem}. From the reduction procedure described at the beginning of part \textbf{b} of the proof, set $\mathrm{A}_\beta$ is reduced to the subset of the $L$-qubit hypercube's vertices
\begin{equation*}
\tilde\mathbb{A}=\Bigg{\{}\left(\frac{1}{2},\frac{1}{2},-\frac{1}{2}\right),\left(\frac{1}{2},-\frac{1}{2},\frac{1}{2}\right),\left(-\frac{1}{2},\frac{1}{2},\frac{1}{2}\right)\Bigg{\}}.
\end{equation*}
The closest to zero point of $\mathrm{conv}(\tilde\mathbb{A})$ is again $\overline{\tilde\beta}=\left(\frac{1}{6},\frac{1}{6},\frac{1}{6}\right)=\sum_{\overline\alpha\in\tilde\mathbb{A}}\frac{1}{3}\overline\alpha$. Therefore, by equation (\ref{eq:psi2}), a critical state from the preimage of $\overline{\tilde\beta}$ is equal to $\ket{\phi_W}$, as one could expect. The final step is to construct such set of diagonal local unitary matrices $k\in T$ that $\bk{\phi_W}{k.\phi_W}=0$. By equation (\ref{eq:perp_cond}) it is equivalent to the problem of finding such angles $\gamma_i,\ i=1,2,3$ that complex numbers $z_j=\frac{1}{3}\e^{i\gamma_j}$ add up to zero. One of the possible solutions is $\gamma_1=0,\ \gamma_2=\frac{2\pi}{3},\ \gamma_3=\frac{4\pi}{3}$, which corresponds to the angles in the equilateral triangle (see Figure \ref{fig:polygon}). Then $\ket{k.\phi_W}=\frac{1}{\sqrt{3}}\left(\ket{110}+e^{\frac{2\pi i}{3}}\ket{101}+e^{\frac{4\pi i}{3}}\ket{011}\right)$ and, finally from equation (\ref{eq:phi_state}), a critical state from $\mu^{-1}(\beta)$ is of the form
\begin{equation*}
\fl\ket{\phi}=\frac{1}{\sqrt{6}}\left[(\ket{110}+\ket{101}+\ket{011})\otimes\ket{0}+(\ket{110}+e^{\frac{2\pi i}{3}}\ket{101}+e^{\frac{4\pi i}{3}}\ket{011})\otimes\ket{1}\right].
\end{equation*}

\section{Summary}
The geometric interpretation of the linear entropy as a squared norm of the momentum map stemming from
the action of the local unitary group enabled us to construct an easily implementable algorithm for finding maximally entangled pure many-qubit states. The main role in this construction is played by the fact that the critical points of $||\mu||^2$ are the $K$-orbits through such $L$-qubit critical states of $||\mu_T||^2$ whose images under both moment maps are the same, i.e. whose one-qubit reduced density matrices are diagonal. The algorithm is in essence a description of collections of $2\times2$ diagonal density matrices, identified with euclidean space $\mathbb{R}^L$, that are such images. We now give an explicit form of the algorithm.

\begin{table}[H]
\centering
\begin{tabular}{lll}
\hline
\multicolumn{3}{c}{\rule[-0.2cm]{0pt}{0.6cm} \textbf{The algorithm for finding critical values of the linear entropy}}\\
\hline
\multicolumn{1}{r}{Input:} & \multicolumn{2}{l}{\rule[-0.2cm]{0pt}{0.6cm}The number of qubits, $L$} \\
\multicolumn{1}{r}{Output:} & \multicolumn{2}{l}{\rule[-0.2cm]{0pt}{0.6cm}Set of critical values of the linear entropy} \\
\hline
\hline
\multirow{2}{*}{\textbf{1}} & \multicolumn{2}{l}{\rule[-0.cm]{0pt}{0.4cm}Construct a list of the hypercube's vertices, i.e. all vectors of length $L$}   \\
 & \multicolumn{2}{l}{whose elements are equal to $\pm \frac{1}{2}$.} \\
\textbf{2} & \multicolumn{2}{l}{For $k\in\{2,\dots,L\}$:} \\
 & \multirow{2}{*}{\textbf{2a}} & Choose a set of $k$ linearly independent vertices of the hypercube, \\
 & & $\{\overline\alpha^{(1)},\dots,\overline\alpha^{(k)}\}$. \\
 & \textbf{2b} & Solve the set of equations (\ref{main_eq}) to obtain coefficients $a_1^c,\dots,a_k^c$ and $\lambda$. \\
 & \textbf{2c} & Construct vector $\overline\beta_c=\sum_i a_i^c\overline\alpha^{(i)}$. \\
 & \multirow{3}{*}{\textbf{2d}} & If all elements of $\overline\beta_c$ are non-negative and $\overline{\beta}_{c}\neq\left(\frac{1}{2},\frac{1}{2},\dots,\frac{1}{2},0\right)$, up to  \\
 & & permutations of vector components, $E_c=\frac{1}{2}+\frac{\lambda}{L}$ is a critical value \\
 & & of the linear entropy. \\
 & \textbf{2e} & Go to 2a, unless all combinations of $k$ vertices have been used. \\
\textbf{3} &  \multicolumn{2}{l}{Add $E_c=0$ corresponding to $\overline\beta_c=\left(\frac{1}{2},\frac{1}{2},\dots,\frac{1}{2}\right)$ to the list of critical values.} \\
\textbf{4} & \multicolumn{2}{l}{Return the list of critical values.} \\
\hline
\end{tabular}
\end{table}

An exemplary state from the preimage of each point generated by this algorithm can be constructed, following the proof of theorem \ref{z_beta_theorem}. In general, the preimages under $\mu$ of collections of diagonal one-qubit reduced density matrices are multidimensional varieties, whose dimensions we calculate in \cite{MOS}. In Figure \ref{fig:calculations} we present the results of calculations for numbers of qubits up to $7$. The number of critical values increases super-exponentially with $L$, as can be easily seen from the plot.

\begin{figure}[H]
\centering
\subfloat{\includegraphics[width=.75\textwidth]{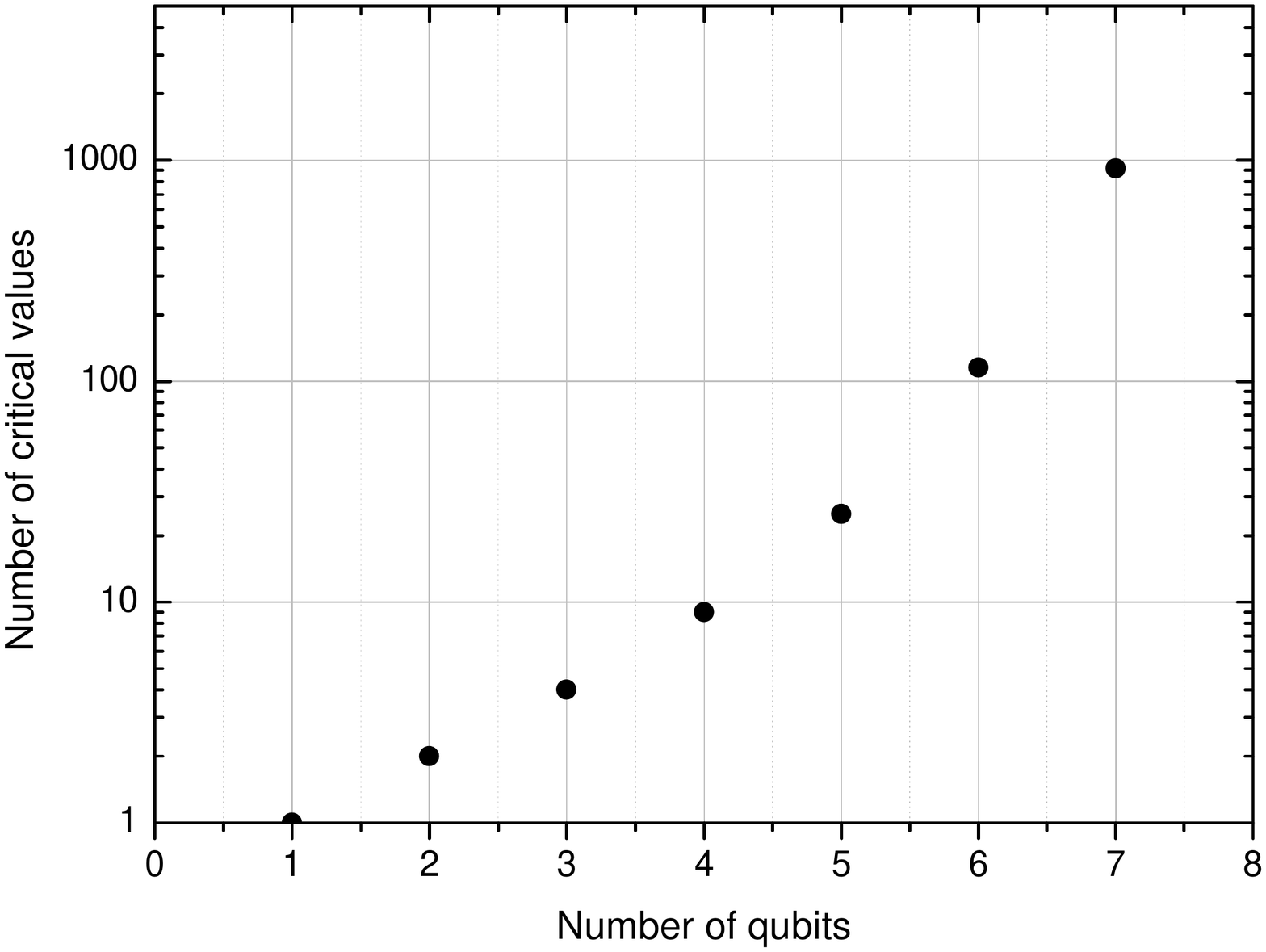}} \\
\subfloat{
\begin{tabular}{c|ccccccc}
$L$ & 1 & 2 & 3 & 4 & 5 & 6 & 7 \\
\hline
number of critical values & 1 & 2 & 4 & 9  & 25  & 115 & 921 \\
\end{tabular}}
\caption{Results of calculations for $L\leq7$. The plot indicates that the number of critical values and therefore the numerical complexity of the algotithm increases superexponentially with $L$.}
\label{fig:calculations}
\end{figure}

The presented algorithm reveals a close relation of the linear entropy critical points to the problem of classifying hyperplanes in hypercube \cite{hyperplanes}. Namely, each minimal combination of weights is a vector perpendicular to a hyperplane, which contains some vertices of the hypercube. For example, the minimal combinations of weights that are proportional to vector $(1,1,\dots,1)$ correspond to hyperplanes whose intersection with the hypercube coincides with the convex hull of the vertices they contain. As explained in \cite{hyperplanes} one encounters fundamentally new types of hyperplanes with the rise of hypercube's dimension. Therefore, the problem of developing general scheme of enumeration of the hyperplanes is difficult and remains still open.

\begin{table}[H]
\centering
\begin{tabular}{c|ccc}
\multicolumn{1}{c}{number of qubits} & $||\overline{\beta}_{c}||^{2}$ & $\overline{\beta}_{c}$ & spanning vertices\\
\hline 
\hline 
\multirow{2}{*}{3} & \multirow{2}{*}{$\frac{1}{12}$} & \multirow{2}{*}{$\left(\frac{1}{6},\frac{1}{6},\frac{1}{6}\right)$} &\sr $\left(-\frac{1}{2},\frac{1}{2},\frac{1}{2}\right)$,$\left(\frac{1}{2},-\frac{1}{2},\frac{1}{2}\right)$\\
 &  &  &\sr $\left(\frac{1}{2},\frac{1}{2},-\frac{1}{2}\right)$\\
\hline 
\multirow{6}{*}{4} & \multirow{2}{*}{$\frac{1}{28}$} & \multirow{2}{*}{$\left(\frac{1}{7},\frac{1}{14},\frac{1}{14},\frac{1}{14}\right)$} &\sr $\left(-\frac{1}{2},\frac{1}{2},\frac{1}{2},\frac{1}{2}\right)$,$\left(\frac{1}{2},-\frac{1}{2},-\frac{1}{2},\frac{1}{2}\right)$\\
 &  &  &\sr $\left(\frac{1}{2},-\frac{1}{2},\frac{1}{2},-\frac{1}{2}\right)$,$\left(\frac{1}{2},\frac{1}{2},-\frac{1}{2},-\frac{1}{2}\right)$\\
\cline{2-4} 
 & \multirow{2}{*}{$\frac{1}{10}$} & \multirow{2}{*}{$\left(\frac{1}{5},\frac{1}{5},\frac{1}{10},\frac{1}{10}\right)$} &\sr $\left(-\frac{1}{2},\frac{1}{2},\frac{1}{2},\frac{1}{2}\right)$,$\left(\frac{1}{2},-\frac{1}{2},\frac{1}{2},\frac{1}{2}\right)$\\
 &  &  &\sr $\left(\frac{1}{2},\frac{1}{2},-\frac{1}{2},-\frac{1}{2}\right)$\\
\cline{2-4} 
 & \multirow{2}{*}{$\frac{1}{4}$} & \multirow{2}{*}{$\left(\frac{1}{4},\frac{1}{4},\frac{1}{4},\frac{1}{4}\right)$} &\sr $\left(-\frac{1}{2},\frac{1}{2},\frac{1}{2},\frac{1}{2}\right)$,$\left(\frac{1}{2},-\frac{1}{2},\frac{1}{2},\frac{1}{2}\right)$\\
 &  &  &\sr $\left(\frac{1}{2},\frac{1}{2},-\frac{1}{2},\frac{1}{2}\right)$,$\left(\frac{1}{2},\frac{1}{2},\frac{1}{2},-\frac{1}{2}\right)$\\
\hline 
\multirow{28}{*}{5} & \multirow{3}{*}{$\frac{1}{76}$} & \multirow{3}{*}{$\left(\frac{3}{38},\frac{1}{19},\frac{1}{19},\frac{1}{38},\frac{1}{38}\right)$} &\sr $\left(-\frac{1}{2},\frac{1}{2},\frac{1}{2},-\frac{1}{2},\frac{1}{2}\right)$,$\left(-\frac{1}{2},\frac{1}{2},\frac{1}{2},\frac{1}{2},-\frac{1}{2}\right)$\\
 &  &  &\sr $\left(\frac{1}{2},-\frac{1}{2},-\frac{1}{2},\frac{1}{2},\frac{1}{2}\right)$,$\left(\frac{1}{2},-\frac{1}{2},\frac{1}{2},-\frac{1}{2},-\frac{1}{2}\right)$\\
 &  &  &\sr $\left(\frac{1}{2},\frac{1}{2},-\frac{1}{2},-\frac{1}{2},-\frac{1}{2}\right)$\\
\cline{2-4} 
 & \multirow{3}{*}{$\frac{1}{52}$} & \multirow{3}{*}{$\left(\frac{3}{26},\frac{1}{26},\frac{1}{26},\frac{1}{26},\frac{1}{26}\right)$} &\sr $\left(-\frac{1}{2},\frac{1}{2},\frac{1}{2},\frac{1}{2},\frac{1}{2}\right)$,$\left(\frac{1}{2},-\frac{1}{2},-\frac{1}{2},-\frac{1}{2},\frac{1}{2}\right)$\\
 &  &  &\sr $\left(\frac{1}{2},-\frac{1}{2},-\frac{1}{2},\frac{1}{2},\frac{1}{2}\right)$,$\left(\frac{1}{2},-\frac{1}{2},\frac{1}{2},-\frac{1}{2},-\frac{1}{2}\right)$\\
 &  &  &\sr $\left(\frac{1}{2},\frac{1}{2},-\frac{1}{2},-\frac{1}{2},-\frac{1}{2}\right)$\\
\cline{2-4} 
 & \multirow{3}{*}{$\frac{1}{44}$} & \multirow{3}{*}{$\left(\frac{1}{11},\frac{1}{11},\frac{1}{22},\frac{1}{22},\frac{1}{22}\right)$} &\sr $\left(-\frac{1}{2},\frac{1}{2},\frac{1}{2},-\frac{1}{2},\frac{1}{2}\right)$,$\left(-\frac{1}{2},\frac{1}{2},\frac{1}{2},\frac{1}{2},-\frac{1}{2}\right)$\\
 &  &  &\sr $\left(\frac{1}{2},-\frac{1}{2},\frac{1}{2},\frac{1}{2},-\frac{1}{2}\right)$,$\left(\frac{1}{2},-\frac{1}{2},-\frac{1}{2},\frac{1}{2},\frac{1}{2}\right)$\\
 &  &  &\sr $\left(\frac{1}{2},\frac{1}{2},-\frac{1}{2},\frac{1}{2},-\frac{1}{2}\right)$\\
\cline{2-4} 
 & \multirow{2}{*}{$\frac{1}{26}$} & \multirow{2}{*}{$\left(\frac{2}{13},\frac{1}{13},\frac{1}{13},\frac{1}{26},\frac{1}{26}\right)$} &\sr $\left(-\frac{1}{2},\frac{1}{2},\frac{1}{2},\frac{1}{2},\frac{1}{2}\right)$,$\left(\frac{1}{2},-\frac{1}{2},-\frac{1}{2},\frac{1}{2},\frac{1}{2}\right)$\\
 &  &  &\sr $\left(\frac{1}{2},-\frac{1}{2},\frac{1}{2},-\frac{1}{2},-\frac{1}{2}\right)$,$\left(\frac{1}{2},\frac{1}{2},-\frac{1}{2},-\frac{1}{2},-\frac{1}{2}\right)$\\
\cline{2-4} 
 & \multirow{2}{*}{$\frac{1}{20}$} & \multirow{2}{*}{$\left(\frac{1}{10},\frac{1}{10},\frac{1}{10},\frac{1}{10},\frac{1}{10}\right)$} &\sr $\left(-\frac{1}{2},\frac{1}{2},\frac{1}{2},-\frac{1}{2},\frac{1}{2}\right)$,$\left(-\frac{1}{2},\frac{1}{2},\frac{1}{2},\frac{1}{2},-\frac{1}{2}\right)$\\
 &  &  &\sr $\left(\frac{1}{2},-\frac{1}{2},-\frac{1}{2},\frac{1}{2},\frac{1}{2}\right)$,$\left(\frac{1}{2},\frac{1}{2},\frac{1}{2},-\frac{1}{2},-\frac{1}{2}\right)$\\
\cline{2-4} 
 & \multirow{3}{*}{$\frac{1}{16}$} & \multirow{3}{*}{$\left(\frac{3}{16},\frac{1}{8},\frac{1}{16},\frac{1}{16},\frac{1}{16}\right)$} &\sr $\left(-\frac{1}{2},\frac{1}{2},\frac{1}{2},\frac{1}{2},\frac{1}{2}\right)$,$\left(\frac{1}{2},-\frac{1}{2},-\frac{1}{2},\frac{1}{2},\frac{1}{2}\right)$\\
 &  &  &\sr $\left(\frac{1}{2},-\frac{1}{2},\frac{1}{2},-\frac{1}{2},\frac{1}{2}\right)$,$\left(\frac{1}{2},-\frac{1}{2},\frac{1}{2},\frac{1}{2},-\frac{1}{2}\right)$\\
 &  &  &\sr $\left(\frac{1}{2},\frac{1}{2},-\frac{1}{2},-\frac{1}{2},-\frac{1}{2}\right)$\\
\cline{2-4} 
 & \multirow{2}{*}{$\frac{3}{28}$} & \multirow{2}{*}{$\left(\frac{3}{14},\frac{3}{14},\frac{1}{14},\frac{1}{14},\frac{1}{14}\right)$} &\sr $\left(-\frac{1}{2},\frac{1}{2},\frac{1}{2},\frac{1}{2},\frac{1}{2}\right)$,$\left(\frac{1}{2},-\frac{1}{2},\frac{1}{2},\frac{1}{2},\frac{1}{2}\right)$\\
 &  &  &\sr $\left(\frac{1}{2},\frac{1}{2},-\frac{1}{2},-\frac{1}{2},-\frac{1}{2}\right)$\\
\cline{2-4} 
 & \multirow{2}{*}{$\frac{1}{8}$} & \multirow{2}{*}{$\left(\frac{1}{8},\frac{1}{8},\frac{1}{8},\frac{1}{4},\frac{1}{8}\right)$} &\sr $\left(-\frac{1}{2},\frac{1}{2},\frac{1}{2},\frac{1}{2},-\frac{1}{2}\right)$,$\left(\frac{1}{2},-\frac{1}{2},-\frac{1}{2},\frac{1}{2},\frac{1}{2}\right)$\\
 &  &  &\sr $\left(\frac{1}{2},\frac{1}{2},\frac{1}{2},-\frac{1}{2},\frac{1}{2}\right)$\\
\cline{2-4} 
 & \multirow{3}{*}{$\frac{9}{44}$} & \multirow{3}{*}{$\left(\frac{3}{11},\frac{3}{11},\frac{3}{22},\frac{3}{22},\frac{3}{22}\right)$} &\sr $\left(-\frac{1}{2},\frac{1}{2},\frac{1}{2},\frac{1}{2},\frac{1}{2}\right)$,$\left(\frac{1}{2},-\frac{1}{2},\frac{1}{2},\frac{1}{2},\frac{1}{2}\right)$\\
 &  &  &\sr $\left(\frac{1}{2},\frac{1}{2},-\frac{1}{2},-\frac{1}{2},\frac{1}{2}\right)$,$\left(\frac{1}{2},\frac{1}{2},-\frac{1}{2},\frac{1}{2},-\frac{1}{2}\right)$\\
 &  &  &\sr $\left(\frac{1}{2},\frac{1}{2},\frac{1}{2},-\frac{1}{2},-\frac{1}{2}\right)$\\
\cline{2-4} 
 & \multirow{2}{*}{$\frac{2}{7}$} & \multirow{2}{*}{$\left(\frac{2}{7},\frac{2}{7},\frac{2}{7},\frac{1}{7},\frac{1}{7}\right)$} &\sr $\left(-\frac{1}{2},\frac{1}{2},\frac{1}{2},\frac{1}{2},\frac{1}{2}\right)$,$\left(\frac{1}{2},-\frac{1}{2},\frac{1}{2},\frac{1}{2},\frac{1}{2}\right)$\\
 &  &  &\sr $\left(\frac{1}{2},\frac{1}{2},\frac{1}{2},-\frac{1}{2},\frac{1}{2}\right)$,$\left(\frac{1}{2},\frac{1}{2},\frac{1}{2},-\frac{1}{2},-\frac{1}{2}\right)$\\
\cline{2-4} 
 & \multirow{3}{*}{$\frac{9}{20}$} & \multirow{3}{*}{$\left(\frac{3}{10},\frac{3}{10},\frac{3}{10},\frac{3}{10},\frac{3}{10}\right)$} &\sr $\left(-\frac{1}{2},\frac{1}{2},\frac{1}{2},\frac{1}{2},\frac{1}{2}\right)$,$\left(\frac{1}{2},-\frac{1}{2},\frac{1}{2},\frac{1}{2},\frac{1}{2}\right)$\\
 &  &  &\sr $\left(\frac{1}{2},\frac{1}{2},-\frac{1}{2},\frac{1}{2},\frac{1}{2}\right)$,$\left(\frac{1}{2},\frac{1}{2},\frac{1}{2},-\frac{1}{2},\frac{1}{2}\right)$\\
 &  &  &\sr $\left(\frac{1}{2},\frac{1}{2},\frac{1}{2},\frac{1}{2},-\frac{1}{2}\right)$\\
\hline 
\end{tabular}
\caption{\label{tabelka2}The set $\mathcal{B}$ for $3$, $4$ and $5$ qubits }
\end{table}

\ack
We would like to thank Marek Ku\'s for encouragement and Valdemar Tsanov for many fruitful discussions at the Langeoog workshop of SFB/TR 12 in 2013 that referred us to Kirwan's method used in this paper. Tomasz Maci\c{a}\.{z}ek is supported by Polish Ministry of Science and Higher Education ``Diamentowy Grant'' no. DI2013 016543 and National Science Center grant no. NCN DEC-2011/01/M/ST2/00379. Adam Sawicki is supported by the Marie Curie International Outgoing Fellowships. We also gratefully acknowledge the support of SFB/TR12 Symmetries and Universality in Mesoscopic Systems program of the Deutsche Forschungsgemeischaft, ERC grant QOLAPS and the support of PL-Grid Infrastructure (http://www.plgrid.pl/en).

\end{document}